\documentclass[journal,onecolumn]{IEEEtranTIE}

\usepackage{amsmath,amsfonts}
\usepackage{algorithmic}
\usepackage{algorithm}
\usepackage{array}
\usepackage[caption=false,font=normalsize,labelfont=sf,textfont=sf]{subfig}
\usepackage{textcomp}
\usepackage{stfloats}
\usepackage{url}
\usepackage{verbatim}
\usepackage{graphicx}
\usepackage{cite}
\usepackage{moreverb,url}
\usepackage[colorlinks,bookmarksopen,bookmarksnumbered,citecolor=red,urlcolor=red]{hyperref}
\usepackage{arydshln}
\usepackage{mathrsfs,fullpage,bm}
\usepackage{dsfont}
\usepackage{float}
\usepackage{bbm}
\usepackage[latin1]{inputenc}
\usepackage{amsmath}
\usepackage{graphicx}
\usepackage{amssymb}
\usepackage{amsthm}
\usepackage{verbatim}
\usepackage{epsfig}
\usepackage[labelfont=bf]{caption}
\usepackage{enumerate}
\usepackage{circuitikz}
\usetikzlibrary{calc}
\usepackage{subfig}
\usepackage{epstopdf}
\usepackage{cancel} %
\usepackage{tikz} %
\usepackage{subcaption}

\usetikzlibrary{calc} %
\makeatletter
\newcommand{\pushright}[1]{\ifmeasuring@#1\else\omit\hfill$\displaystyle#1$\fi\ignorespaces}
\makeatother
\newtheorem{fact}{Fact}
\newtheorem{definition}{Definition}
\newtheorem{assumption}{Assumption}
\newtheorem{lemma}{Lemma}
\newtheorem{remark}{Remark}

\newtheorem{proposition}{Proposition}

\usepackage{tikz}

\def\begquo{\begin{quote}}
	\def\endquo{\end{quote}}
\def\begequarr{\begin{eqnarray}}
	\def\endequarr{\end{eqnarray}}
\def\begequarrs{\begin{eqnarray*}}
	\def\endequarrs{\end{eqnarray*}}
\def\begarr{\begin{array}}
	\def\endarr{\end{array}}
\def\begequ{\begin{equation}}
	\def\endequ{\end{equation}}
\def\lab{\label}
\def\begdes{\begin{description}}
	\def\enddes{\end{description}}
\def\begenu{\begin{enumerate}}
	\def\begite{\begin{itemize}}
		\def\endite{\end{itemize}}
	\def\endenu{\end{enumerate}}

\def\lef[{\left[\begin{array}}
	\def\rig]{\end{array}\right]}
\def\qed{\hfill$\Box \Box \Box$}
\def\begcen{\begin{center}}
	\def\endcen{\end{center}}
\def\begrem{\begin{remark}\rm}
	\def\endrem{\end{remark}}
\def\begdef{\begin{definition}}
	\def\enddef{\end{definition}}
\def\begpro{\begin{proposition}}
	\def\endpro{\end{proposition}}
\def\begfac{\begin{fact}}
	\def\endfac{\end{fact}}
\def\begass{\begin{assumption}}
	\def\endass{\end{assumption}}

\def\begsubequ{\begin{subequations}}
	\def\endsubequ{\end{subequations}}

\def\begmat#1{\begin{bmatrix}#1\end{bmatrix}}
\def\begali#1{\begin{align}{#1}\end{align}}
\def\begalis#1{\begin{align*}{#1}\end{align*}}

\def\PP{{\mathsf  P}}
\def\RR{{\mathsf  R}}
\def\HH{{\mathsf  h}}

\def\L2e{{\cal L}_{2e}}
\def\bul{\noindent $\bullet\;\;$}
\def\rea{\mathds{R}}

\def\hal{{1 \over 2}}

\def\min{{\mbox{min}}}

\def\hal{{1 \over 2}}

\usepackage{color}

\usepackage[prependcaption,colorinlistoftodos]{todonotes}

\graphicspath{{}{img/}}

\begin{document}

\title{A New Partial State-Feedback IDA-PBC for Two-Dimensional Nonlinear Systems: Application to Power Converters with Experimental Results}

\author{
	\vskip 1em
	
	Rafael~Cisneros, 
	Leyan~Fang, 
    Wei~He, 
	and Romeo~Ortega, 

	\thanks{
	R. Cisneros, L. Fang and R. Ortega are with the Department of Electrical and Electronic Engineering, ITAM, Mexico City, 01080 Mexico (email: rcisneros@itam.mx; leyan.fang@itam.mx; romeo.ortega@itam.mx).
	
L. Fang is also with the Center for Control Theory and Guidance Technology, Harbin Institute of Technology, Harbin
150001, China.

W. He is with the School of Automation, Nanjing University of Information Science \& Technology, 210044  Nanjing, China (email: hwei@nuist.edu.cn).
	}
}

\maketitle

\begin{abstract}
 In this paper we propose a variation of the widely popular Interconnection-and-Damping-Assigment Passivity-Based Control (IDA-PBC) based on Poincare's Lemma to design {\em output feedback} globally stabilizing controllers for two dimensional systems. The procedure is constructive and, in comparison with the classical IDA-PBC,  whose application is often stymied by the need to solve the (infamous) matching {\em partial} differential equation (PDE), in this new method the PDE is replaced by an {\em ordinary} differential equation, whose solution is far simpler. The procedure is then applied  for the design of {\em voltage-feedback} controllers for the three most typical DC-to-DC power converter topologies: the Buck, Boost and Buck-Boost. It is assumed that these converters feed an {\em uncertain} load, which is characterized by a static relation between its voltage and current. In the case when the load consists of the parallel connection of a resistive term and a constant power load we propose an {\em adaptive} version of the design, adding an identification scheme for the load parameters. This allows the controller to regulate the converter output when the load varies---that is a typical scenario in these applications.  Extensive numerical simulations and {\em experimental} results validate the approach.
\end{abstract}
%
\section{Introduction and Background on IDA-PBC}
\lab{sec1}
%
\IEEEPARstart{S}{tabilization} of physical systems by shaping their energy function is a well established technique whose roots date back to the work of Lagrange and Legendre.  Potential energy shaping for fully actuated mechanical systems was first introduced, more than 40 years ago, in \cite{TAKARI}. In \cite{ORTSPO} it was proved that passivity was the key property underlying the stabilisation mechanism of these designs and the, now widely popular, term of {\em passivity-based control} (PBC) was coined---to which many books have been devoted \cite{ORTetalbook,ORTetalbookpid, HATANAKAbook,BAIbook,BAObook,SERRAbook}.

A variation of PBC that has been very successful in many practical applications of equilibrium stabilization is {\em Interconnection and Damping Assignment} (IDA-PBC), first reported in \cite{ORTetal}, and further ellaborated in \cite{ORTGAR}. The main idea of IDA-PBC is to, via feedback, impose to the closed-loop dynamics a {\em port-Hamiltonian} (pH) form \cite{VANbook}. That is, that the closed-loop takes the form 
$$
\dot{x}= Q(x) \nabla P(x)
$$
where $x(t) \in \rea^n$ is the system state, the mapping $Q: \rea^n \to \rea^{n \times n}$ verifies
$$
Q(x)+Q^\top(x)\leq 0
$$ 
and the scalar function $P: \rea^n \to \rea$ satisfies
\begequ
\lab{minhd}
x_\star  = \arg \min \{P(x)\},
\endequ
where $x_\star \in \rea^n$ is the desired equilibrium to be stabilized. It is well-known that the main stumbling block for the application of IDA-PBC is that, as all constructive controller designs, require the solution of a {\em partial differential equation} (PDE)---called the matching PDE---a task that is often complicated. 

In \cite{ORTGAR} an interesting variation of IDA-PBC was reported. It relies on the application of {\em Poincare's Lemma} \cite[Theorem 10.39]{RUDbook}---a simpler version is given  in Appendix \ref{appa}. Applied to nonlinear systems of the form $\dot x=f(x,u)$, and assuming that the matrix $Q(x)$ above is {\em full rank}, Poincare's Lemma states that a necessary and sufficient condition for the existence of $P(x)$, a solution of the matching PDE,
$$
\nabla P(x)=Q(x)^{-1} f(x,\hat u(x))
$$ 
with $\hat u:\rea^n \to \rea^m$ the state feedback control signal, is that the right-hand side satisfies
$$
\nabla \{Q(x)^{-1} f(x,\hat u(x))\}=\Big(\nabla \{Q(x)^{-1} f(x,\hat u(x))\}\Big)^\top.
$$
It turns out that the symmetry condition above can, in some cases, be translated into an {\em ordinary} differential equation (ODE) for the control signal, whose solution is far simpler than the matching PDE.

This variation of the IDA-PBC method has received a scant attention. It was used in \cite{RODetal} to design a very simple static nonlinear output-feedback controller for the Boost converter. It was explored in \cite{ZHAetal} for the design of general state-feedback controllers and in \cite{LIUORTSU} for the control of chemical processes. The main objective of this paper is to further explore the application of this IDA-PBC methodology.

The main goals of the paper are:

\begenu[{\bf G1}]
\item We provide a {\em constructive} version of IDA-PBC for the case of {\em output feedback} controller designs for second order systems, which are {\em non linear} in the input signal. 
\item The method does not require a solution of a PDE, instead it needs to solve an ODE a task which, as shown in the paper, is sometimes obvious.
\item We apply the method for the design of {\em voltage-feedback} controllers for the three most popular DC-to-DC power converters. The main feature of the resulting design is that we don't need to specify the nature of the converter {\em load}, which is simply described by a static relation between the voltage and the current of the load.
\item If the load takes the practically relevant form of a linear resistor in parallel with a {\em constant power} load, we derive a Lyapunov function for the closed-loop system, that is used to estimate the {\em domain of attraction} of the desired equilibrium. 
\item For the case of the load described above, we design an {\em adaptive} version of the controller that estimates---with {\em finite convergence time} (FCT)-- the linear resistor and the power of the constant power load, allowing the possibility of applying the controller for the often encountered case of uncertain and time-varying loads.
\item The applicability of the proposed controller is validated with comprehensive {\em experimental} evidence on the power converters.   
\endenu

The rest of the paper has the following structure. The main theoretical result, that is the proposed IDA-PBC design method, is presented in Section \ref{sec2}. Section \ref{sec3}  is devoted to the description of the mathematical models of the converters studied in the paper and the derivation of their assignable equilibrium sets.   The second main result of the paper---that is, the application of the controller design procedure to the power converters---is given in Section \ref{sec4}. In Section \ref{sec5} we present the adaptive version of the controller. Section \ref{sec6} presents some simulation results and the experimental evidence is given in Section \ref{sec7}. We wrap-up the paper with some concluding remarks and future research in Section \ref{sec8}. \\

\noindent {\bf Notation} All mappings are assumed {\em smooth}. For a scalar function $D \colon \mathbb{R}^{2} \to \mathbb{R}$, we define  $\nabla D := {\left({\partial D \over \partial x}\right)^{\top}},$ $\nabla_{x_i} D := {\left({\partial D \over \partial x_i}\right)^{\top}},i=1,2,$ and  $\nabla^2 D := {\left({\partial^2 D \over \partial x^2}\right)^{\top}}$. For a function of scalar argument $g:\rea \to \rea$, we define $g'(z):={dg(z) \over dz}$. Given the function $\hat u: \rea \to \rea$ and any mapping $Q: \rea^2 \times \rea \to \rea^{m \times n}$,  we define the mapping {\em composition}
\begequ
\lab{hatm}
\hat Q: \rea^2  \to \rea^{m \times n},\;\hat Q(x) :=  Q(x,\hat u(x_2)).
\endequ
{Given a map $F(x)$, we define for the distinguished {\em constant} element $x_\star  \in \rea^n$}, the {\em constant} matrix $F_\star  :=  F(x_\star)$. 
The arguments of a function are omitted when clear from the context.  
%
\section{Proposed Poincare's Lemma-Based IDA-PBC}
\lab{sec2}
%
The main output feedback IDA-PBC design proposed in the paper is contained in the proposition below. Its construction is based on the {\em third variation} of IDA-PBC proposed in \cite{ORTGAR}, which relies on the application of Poincare's Lemma.  In this paper, specializing to output feedback control of two-dimensional systems, we provide a constructive procedure whose main feature is that, in contrast with classical IDA-PBC, it {\em does not} require the solution of the standard matching PDE, which is replaced by an ODE.

\begpro
\lab{pro1}\em
Consider the two-dimensional nonlinear system
\begsubequ
\lab{sys}
\begali{
	\lab{sysdotx}
	\dot x &= f(x,u),\\
	\lab{syss}
	y& = x_2,
}
\endsubequ
with $x(t) \in \rea^2$, $u(t) \in \rea$,  $y(t) \in \rea$ and a desired assignable equilibrium point $x^\star   \in \rea^2$. Assume there exist mappings
\begalis{
	& \alpha_1:\rea^2 \times \rea \to \rea,\; \alpha_2:\rea^2 \times \rea \to \rea,\; \\
   &  \beta:\rea^2 \times \rea \to \rea,\; \hat u:\rea \to \rea,
}
such that the following conditions hold true.
\item[C1.] $\hat \alpha_1(x) \hat \alpha_2(x) + \hat \beta^2(x) \neq 0$. 

\item[C2.] $\hat \alpha_1(x) \leq 0,\;\hat \alpha_2(x) \leq 0$,

\item[C3.] The scalar functions $D_i:\rea^2 \times \rea \to \rea,i=1,2$
\begsubequ
\lab{d}
\begali{
	\lab{d1}
	D_1(x,u) & :=\alpha_1(x,u) f_1(x,u) + \beta(x,u) f_2(x,u) \\
	\lab{d2}
	D_2(x,u) &:= -\beta(x,u) f_1(x,u) + \alpha_2(x,u) f_2(x,u)
}
\endsubequ
when evaluated at $u=\hat u(x)$, satisfy
\begequ
\lab{poilem}
\nabla_{x_2} \hat D_1 = \nabla_{x_1} \hat D_2.
\endequ
\item[C4.] $\hat f_\star =0$

\item[C5.] The following constant matrix condition holds
\begin{align}
	\lab{codp}
	& \begmat{\left(\nabla_{x_1} \hat D_1\right)_\star  & \left(\nabla_{x_2} \hat D_1\right)_\star \\ \left(\nabla_{x_1} \hat D_2\right)_\star  & \left(\nabla_{x_2} \hat D_2\right)_\star }  >0.
\end{align}
\item[C6.]  The largest invariant set contained in the set
\begequ
\label{asysta}
\{ x \in \rea^{2} \mid \hat \alpha_1(x) \hat f^2_1(x) + \hat \alpha_2(x) \hat f^2_2(x)=0 \}
\endequ
equals  $\{ x_\star \}$.
%
Under these conditions $x_\star $ is a {\em globally asymptotically stable} equilibrium of the closed-loop system $\dot x=f(x,\hat u(x_2))$.
\end{proposition}

\begin{proof}
The gist of the proof is to show that, under the conditions of the proposition, the closed-loop system $\dot x=f(x,\hat u(x))$ takes the pH form
\begequ
\lab{phsys}
\dot x=\hat Q^{-1}(x)\nabla P(x),
\endequ
where $\hat Q:\rea^2 \to \rea^{2 \times 2}$ is a full rank matrix, whose symmetric part is negative semidefinite, and $P:\rea^2 \to \rea$ is a positive definite function (with respect to $x_\star$).  That is, the the proposed controller belongs to the class of IDA-PBC \cite{ORTetal,ORTGAR}.

To establish the proof we proceed as follows. First, define the matrix
\begali{
	\lab{q}
	Q(x,u) & :=\begmat{ \alpha_1(x,u) & \beta(x,u)  \\ - \beta(x,u) & \alpha_2(x,u)}.
}
We have the following equivalences:
\begalis{
	\eqref{d} & \Longleftrightarrow D(x,u)=Q(x,u)f(x,u)\\
	{\bf C1} & \Longleftrightarrow \det\{\hat Q(x)\} \neq  0 \Longleftrightarrow \exists\; \hat Q^{-1}(x)\\
	{\bf C2} & \Longleftrightarrow \hat Q(x)+\hat Q^\top(x) \leq 0\\
	{\bf C3} & \Longleftrightarrow  \nabla \hat D(x)=\Big(\nabla \hat D(x)\Big)^\top\\
& \Longleftrightarrow \exists P:\rea^2 \to \rea\;|\; \hat D(x)=\nabla P(x) \\
	{\bf C4} & \Longleftrightarrow x_\star\mbox{\;is\;a\;closed-loop\;equilibrium}\\
	{\bf C5} & \Longleftrightarrow \nabla \hat D_\star(x)>0  \Longleftrightarrow \Big(\nabla^2 P(x)\Big)_\star>0,
}
where the second equivalence in {\bf C3} follows from Poincare's Lemma given in Appendix \ref{appa}. Now, we have
\begin{align}
&\eqref{d}\;(\mbox{evaluated\;at\;}u=\hat u(x))\; \& \;{\bf C1} \nonumber \\
&\Rightarrow \dot x = \hat Q^{-1}(x)\hat D(x) = \hat Q^{-1}(x)\nabla P(x). \nonumber
\end{align}
where we used {\bf C3} to get the second identity, establishing \eqref{phsys}.

The remaining part of the proof follows the standard Lyapunov-based stability analysis of pH systems used in IDA-PBC. Namely, from {\bf C4}, we have that $x_\star $ is an equilibrium of the closed-loop system. The latter, together with {\bf C1}, ensures that  $\nabla P(x_\star )=0$. This together with {\bf C5} ensures that $P(x)$ is positive definite (with respect to $x_\star$). Evaluating $\dot P$ along the trajectories of the closed-loop system \eqref{phsys} yields
\begalis{
	\dot P = &  \nabla^{\top}P(x)\hat Q^{-1}(x)\nabla P(x) \\
	 = & \hal \nabla^{\top}P(x)[\hat Q^{-1}(x)+\hat Q^{-\top}(x)]\nabla P(x) \\
	= &\frac{1}{2} \hat f^\top(x) [\hat Q(x)+\hat Q^\top(x)]\hat f(x)\leq 0
}
where the inequality is obtained invoking again {\bf C2}. This ensures $P(x)$ is a Lyapunov function, establishing global stability. The proof is completed noting that 
\begalis{
\hat f^\top(x) [\hat Q(x)+\hat Q^\top(x)]\hat f(x)&= \hat \alpha_1(x) \hat f^2_1(x) + \hat \alpha_2(x) \hat f^2_2(x),
}
imposing condition {\bf C6} and invoking LaSalle's Theorem \cite[Theorem 4.4]{KHAbook}.
\end{proof}

\begrem
\lab{rem0}
For the purposes of the stability proof, the condition {\bf C1}, which ensures the matrix $\hat Q$ is invertible, {\em may be removed.} Indeed, for any $\hat Q(x)$, \eqref{sysdotx} implies $\hat Q(x)\dot x = \hat Q(x)f(x,u).$  Furthermore, {\bf C3} and  {\bf C4} ensure that 
	$$
		\hat Q(x)f(x,u)=\nabla P(x).
	$$
The proof is completed evaluating $\dot P$, which yields 
$$
\dot P=\frac{1}{2} \hat f^\top(x) [\hat Q(x)+\hat Q^\top(x)]\hat f(x),
$$
that is the same expression we have above. We decided to keep condition {\bf C1} in the main result to make the connection with IDA-PBC, which imposes the closed-loop pH structure \eqref{phsys}. However, notice that its removal is an important relaxation because you can take {\em only one} of the $\alpha_i(x,u)$ to be nonzero, while setting to zero $\beta(x,u)$ and the other one---considerably simplifying \eqref{poilem}.
\endrem

\begrem
\lab{rem1}
It is important to underscore the fact that \eqref{poilem} is evaluated with the control signals $\hat u(x_2)$---simplifying the solution of this equation. This should be contrasted with the equivalent condition in standard IDA-PBC, where it is the (infamous) matching PDE, whose solution usually stymies the application of IDA-PBC.
\endrem 
%
\section{Three DC-to-DC Power Converters}
\lab{sec3}
%
As an illustration of application of  Proposition \ref{pro1} we will design in the next section stabilizing, {\em voltage-feedback}, controllers for the Buck, Boost and Buck-Boost DC-to-DC power converters. Towards this end, we give in the next subsection their mathematical models. Then, to simplify the analytical expressions we propose to use scaled models in Subsection \ref{subsec32} and, finally, in Subsection \ref{subsec33} we present their assignable equilibrium set. 

\subsection{Dynamic models of the converters}
\lab{subsec31}
The electrical circuits of the three studied converters are depicted in Fig. \ref{conv}.  The variables $i$ and $v$ in the figure are, respectively, the inductor current and capacitor voltage in each converter topology. Under normal operation conditions, these variables take \textit{non-negative} values. Moreover, the positive parameters $L,C, G$ and $E$ are, respectively, the converter inductance, capacitance, load conductance and voltage source. Also, $s\in \{0,1\}$ is the position of the switch that acts as control input.  

\begin{figure}[h]
	\centering
	\includegraphics[width=\linewidth]{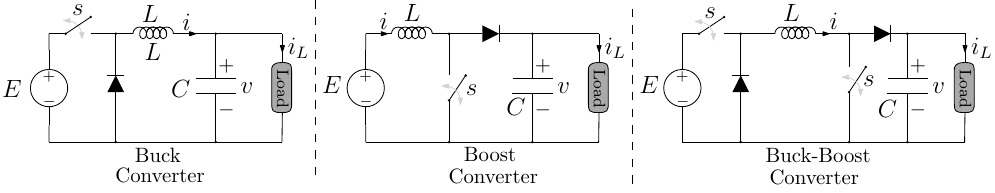}
	\caption{The Buck, Boost and Buck-Boost converters---the signal $s\in\{0,1\}$ opens or closes the converter switch. }
\label{conv}
\end{figure}

A block labelled as ``Load'' is connected at the output port of each converter. One interesting feature of our developments is that we {\em do not need} to give a particular ``structure" to the load, a scenario that is very common in applications where typically the load is {\em uncertain}. We will only assume that there is a {\em static} relation between the {\em current} feeding the load ($i_L$) and the terminal voltage ($v$). That is, there exists a function $\HH:\rea_+ \to \rea_+$ such that
\begequ
\lab{ilhv}
i_L=\HH(v).
\endequ

We introduce, at this point, the so-called \textit{average} models of the three converters, where we defined the {\em continuous} control signal
$$
u:=1-\mbox{avg}\{s\},
$$
which ranges in $u \in(0,1]$---for further details see \cite{ERIMAKbook}. Starting with the Buck converter, its averaged model equations are\footnote{Notice that we are using the symbol $\tau$ to denote the time, which will be modified to $t$ in the sequel when we do a time scaling.}
\begin{equation}\label{conv1}
\begin{aligned}
	L\frac{di}{d\tau} & = -v +uE, ~C\frac{dv}{d\tau} & =  i -\HH(v) .
\end{aligned}
\end{equation}
The equations of the Boost and the Buck-Boost converter can be merged into a single one as follows
\begin{equation}\label{conv2}
\begin{aligned}
	L\frac{di}{d\tau} &= -vu + E+ {\tt g}(u),~ C\frac{dv}{d\tau} &=   iu - \HH(v),
\end{aligned}
\end{equation}
where, the map ${\tt g}:\mathbb{R}\to\mathbb{R}$ is defined as follows
$$
{\tt g}(u)=
\begin{cases}
0\; &\text{for the Boost converter,}\\
-Eu \;&\text{for the Buck-Boost converter}.
\end{cases}
$$
\subsection{A particular representation of the load}
\lab{subsec32}
As will be shown in the sequel, applying Proposition \ref{pro1}, we will derive the control law for an arbitrary load characterized by \eqref{ilhv}. On the other hand, to derive an associated Lyapunov function and to propose an adaptive version of the controller, we need to assume a particular structure for $\HH(v)$. The interest of deriving a Lyapunov function is to be able to determine an estimate of the {\em domain of attraction} of the equilibrium. Furthermore, for this particular load, we delop an {\em adaptive} version of the controller, which allows us to deal with {\em time-varying loads}. 

For the purposes of the derivation of the Lyapunov function and the design of an adaptive controller, we assume the load consists of a resistor in parallel with a constant power load (CPL)---see Fig. \ref{ld}---which is a very general and common scenario.  Applying Kirchhoff's current law we obtain
\begin{align}\label{ym}
i_L= Gv+ \frac{P_{\tt{cpl}}}{v},
\end{align}
where the positive constants $G$ and $P_{\tt{cpl}}$ are, respectively, the parallel resistor admittance the power level of the CPL. 

\begin{figure}[h!]
\centering
\includegraphics[width=0.15\linewidth]{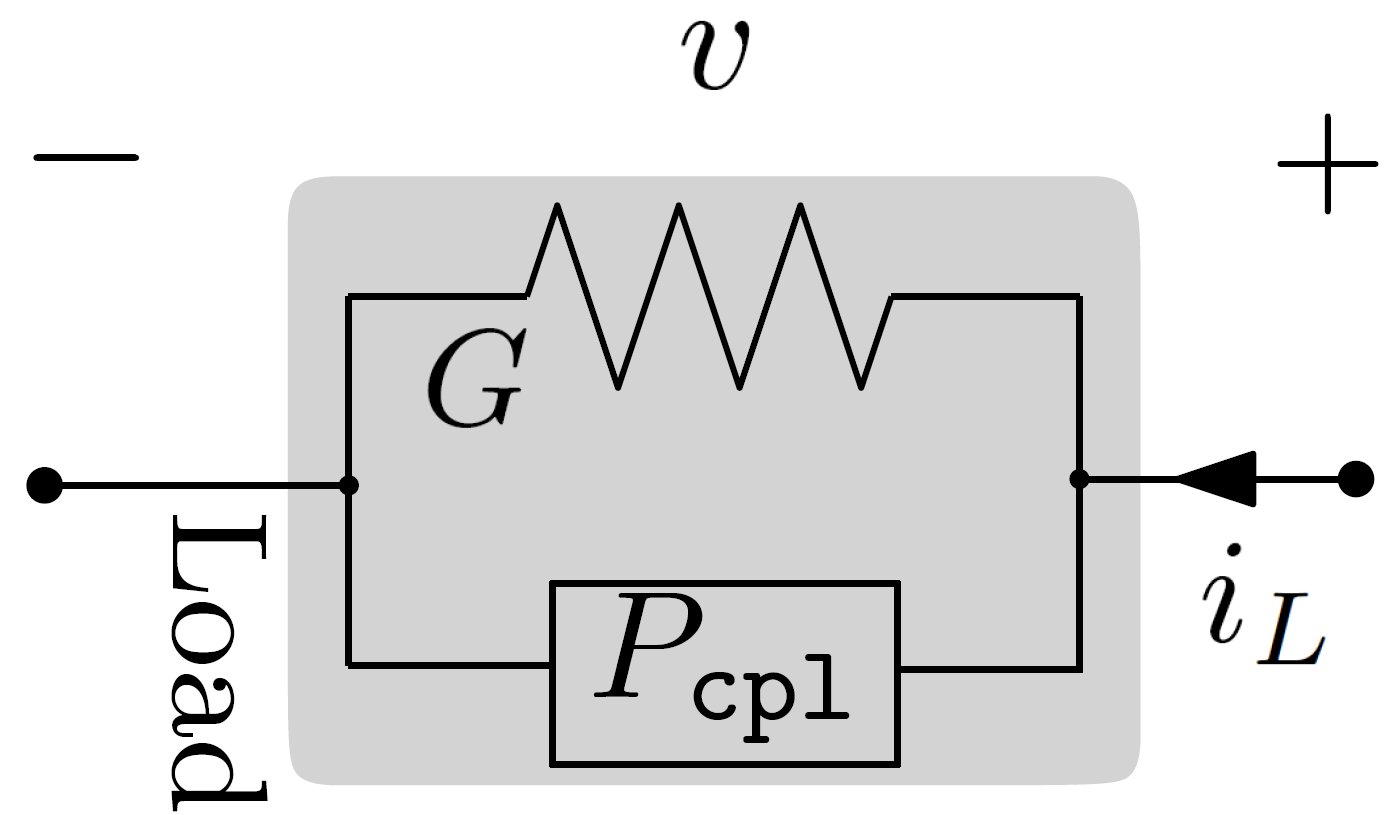}
\caption{The converter load for the Lyapunov function derivation}
\label{ld}
\end{figure}
\subsection{Normalized models of the converters}
\lab{subsec33}
It is well-known that the average models can be normalized by means of both time and variable scaling \cite{SIRSILbook}. The interest of using normalized models is, of course, to simplify the analytic expressions of our controllers, which are presented in the {\bf Fact} below. The proof of the transformation is straightforward, hence it is omitted---the interested reader is referred to \cite{SIRSILbook} for further details.

\begin{fact}\label{F1}\em
Using the standard symbol $\dot {(\cdot)}:=\frac{d(\cdot)}{dt}$, to denote the derivative with respect to the new time variable $t\in\mathbb{R}$, we have that an equivalent  representation of the Buck, Boost and Buck-Boost converters \eqref{conv1}, \eqref{conv2} is given by the following normalized models.\\

\bul Buck converter:
\begin{equation}\label{bck}
\begin{aligned}
	\dot x_1&= -x_2 +u,~\dot x_2=  x_1- h(x_2)
\end{aligned}
\end{equation}
\bul Boost and Buck-Boost converter:
\begin{equation}\label{bst}
\begin{aligned}
	\dot x_1&= -g(x_2)u +1,~ \dot x_2=  x_1u - h(x_2),
\end{aligned}
\end{equation}
where,  $g:\mathbb{R}_+ \to \mathbb{R}_+$ is given by
\begin{equation}\label{gx2}
g(x_2)=\begin{cases}
	x_2\; &\text{for the Boost converter,}\\
	x_2+1 \;&\text{for the Buck-Boost converter},
\end{cases}
\end{equation}
and  $h:\mathbb{R}_+ \to \mathbb{R}_+$ is an arbitrary, differentiable function.
\qed
\end{fact}

\begrem
\lab{rem2}
The variables $x_1$, $x_2$, and $t$ are related to $i$, $v$ and $\tau$ as follows
\begin{align}\label{transf}
x_1= 	\frac{1}{E}\sqrt{\frac{L}{C}} i ,\; x_2= \frac{1}{E}v,\; t=\frac{1}{\sqrt{CL}}\tau.
\end{align}
\endrem
\begrem
\lab{rem3}
For the case when the load is given by \eqref{ym} we have that\footnote{For this particular load, it will be assumed throughout the paper that $x_2(t)$ is bounded away from zero.} 
\begin{equation}\label{hx2}
h(x_2)= \RR x_2 + \frac{\PP}{x_2},
\end{equation}
where
\begin{align}\label{param}
\RR:=G\sqrt{\frac{L}{C}},\;\; \PP:=\frac{P_{\tt{cpl}}}{E^2}\sqrt{\frac{L}{C}}.
\end{align}
\endrem
\subsection{Assignable equilibrium set}
\lab{subsec34}
The control objective is to regulate the variable $x_2$ to some arbitrary value $x_{2\star}>0$. Let  $u_\star$ and $x_\star$ be the state and control equilibrium values such that $x_2=x_{2\star}$. Applying \cite[Proposition B.1]{ORTetalbookpid}, we give below the {\em assignable equilibrium set} with the corresponding---uniquely defined---constant control value for the three converters described above.\\ 

\bul Buck converter:
\begin{align}\label{eqbck}
\mathcal{E}_{\mathrm{bck}}:= \left\{ (x,u)\in\mathbb{R}^3 | x_1=h_\star, x_2=x_{2\star}, u= x_{2\star}  \right\}.
\end{align}

\bul Boost and Buck-Boost converters:
\begin{align}\label{eqbs}
\mathcal{E}_{\mathrm{bb}}:= \left\{ (x,u)\in\mathbb{R}^3 | x_1=g_\star h_\star,x_2=x_{2\star}, u=\frac{1}{g_\star} \right\}.
\end{align}
where the  mapping $g(x_2)$ is defined in \eqref{gx2}. \\

The following ``weak" assumption will be imposed in the sequel. 

\begin{assumption}\label{ass1}\em
The function $h(x_2)$, relating the load current and its voltage, satisfies 
\begequ
\lab{hprista}
h'_\star  >0.
\endequ
\end{assumption}

\begrem \label{condx2}
Setting $h(x_2)$ as in \eqref{hx2}, and then replacing it into \eqref{hprista} results in the following condition
$$
x_{2\star}>\sqrt{\frac{\PP}{\RR}}.
$$
\endrem
%
\section{Control Design of the Power Converters}
\lab{sec4}
%
In this section we apply Proposition \ref{pro1} to design stabilizing, {\em voltage-feedback}, controllers for the Buck, Boost and Buck-Boost DC-to-DC power converters.\footnote{If we don't impose any constraint on the sign of the state components, the controllers are {\em globally} stabilizing. However, we recall that, due to topological constraints, the converter signals are restricted to live in the positive orthant $\mathbb{R}_+^3$.}  It is important to recall that the main difficulty in the control of these converters is that the output signal, that is the voltage fed to the load, has the behavior of a {\em non-minimum phase} output \cite{ISIbook}. Due to this behavior the vast majority of the practical controllers designed for these converters are of the {{\em indirect} type,} where a PI loop is placed around the current, whose reference is determined to match the desired voltage value. This approach is clearly extremely fragile as the derivation of the current reference requires the exact knowledge of the system parameters or the design of an outer control loop for the converter output voltage, which increases the complexity of the whole controller. This issue has been extensively discussed in the literature, see {\em e.g.} \cite[Section 4.3.A]{ORTetalbook}, \cite{SIRSILbook}  and references therein. Interestingly, it has recently been shown in \cite{FANORTGRI} that it is possible to stabilize the Boost converter applying a single PI controller with input the load voltage, provided the PI gains are suitably selected.   
%
\subsection{Buck Converter}
In this section we design the IDA-PBC for the Buck converter.
\subsubsection{Proposed controller} 

\begpro\label{pro2}\em
Consider the Buck converter dynamics \eqref{bck}. Fix $x_{2\star}$ verifying \eqref{ass1}. The mappings 
\begin{equation}
	\begin{aligned}\label{qbck}
		{\alpha}_1(x) =-\frac{1}{k} ,\;
		{\alpha}_2(x)  = 0,\;
		{\beta}(x)  =1,
	\end{aligned}
\end{equation}
and the control signal
\begequ
\lab{ubuc}
\hat u(x_2)=-k[ h(x_2)-x_{1\star}] + x_2
\endequ
with $k>0$ a tuning gain, fulfill the conditions \textbf{C1}-\textbf{C6} of Proposition \ref{pro1}. Consequently, the  system in closed-loop with the control $u=\hat u(x_2)$ has an equilibrium point $(x_{1\star},x_{2\star})\in\mathcal{E}_{\mathrm{bck}}$ which is asymptotically stable.

Moreover, if the load is of the form \eqref{ym}, a Lyapunov function for the equilibrium is given by
\begin{align}
	P(x)=&\frac{1}{2}(x_1-x_{1\star})^2 +{ {{k\RR(x_2-x_{2\star})} \over 2}}
\Big[ x_2+x_{2\star} -\frac{2x_{1\star}}{\RR}\Big]+ k \PP\ln \frac{x_2}{x_{2\star}}. 
\label{Potbck}
\end{align}   
\endpro

\begin{proof}
	The proof  is given in Appendix \ref{appb}.
\end{proof}

\begrem
\lab{rem4}
{Notice that if the condition of {\bf Assumption} \ref{ass1} is  not satisfied, that is, if $h'_\star <0$, we can set the free constant $k<0$ and take $\alpha_1={1 \over k}$. }	
\endrem

\subsubsection{Numerical estimation of the region of attraction}
\lab{subsubsecbc}

Consider the system \eqref{bck} with the load \eqref{ym} in closed loop with \eqref{ubuc}. Consider the circuit parameters given in Table \ref{para}, and evaluate the values of the constants $\RR$ and $\PP$, according to \eqref{param}. Also, set the controller gain $k=0.1$. The desired equilibrium point is $(x_{1\star},x_{2\star})=(0.0285,0.833)$---according to \eqref{transf}, corresponds to an output voltage setpoint of $v_\star=20$V. {From straightforward computations, it can be seen that this value verifies \textbf{Assumption \ref{ass1}}  and \textbf{Remark \ref{condx2}}.}

We will show simulated results of the closed-loop system in the phase plane. Fig. \ref{bckf} shows a set of trajectories, together with some level sets of the Lyapunov function $P(x)$, that is
$$
\Omega_{\bar P}:=\{x \in \rea^2\;|\;P(x) \leq {\bar P} \}
$$
{{with $\bar P \in \{0.000215,~0.00025,0.0003\}$.}} 
 
 Fig. \ref{bckf} (b) displays an enlarged view of the selected area in Fig. \ref{bckf} (a). As it can be seen from this figure, there exists an invariant set---in gray---in the {\em first quadrant} of the plane, such that all trajectories starting in this set remain in this set and converge to the desired equilibrium point---that is, trajectories starting in this set remain physically valid. 

 \begin{table}[!t]
	\renewcommand{\arraystretch}{1.3}
	\caption{Parameters of the Experimental Setup and Simulated System}
	\centering
 \label{para} 
	\resizebox{0.25\columnwidth}{!}{
		\begin{tabular}{l l l}
			\hline\hline \\[-3mm]
			\multicolumn{1}{c}{Parameter} & \multicolumn{1}{c}{Value} & \multicolumn{1}{c}{{Unit}}  \\[1.6ex] \hline
			 $E$ & $24$ & V \\
   $L$  &   $1$& mH             \\
    $C$ &     $330$ & $\mu$F           \\
    {$G $} & ${0.0167}$ & $\mho$\\
    $P_{\tt cpl}$  & $1.2$& W    \\ [1.4ex]
			\hline\hline
		\end{tabular}
	}
\end{table}

\begin{figure}[!ht]
  \centering
  \begin{minipage}[b]{0.45\linewidth}
    \centering
    \includegraphics[width=\linewidth]{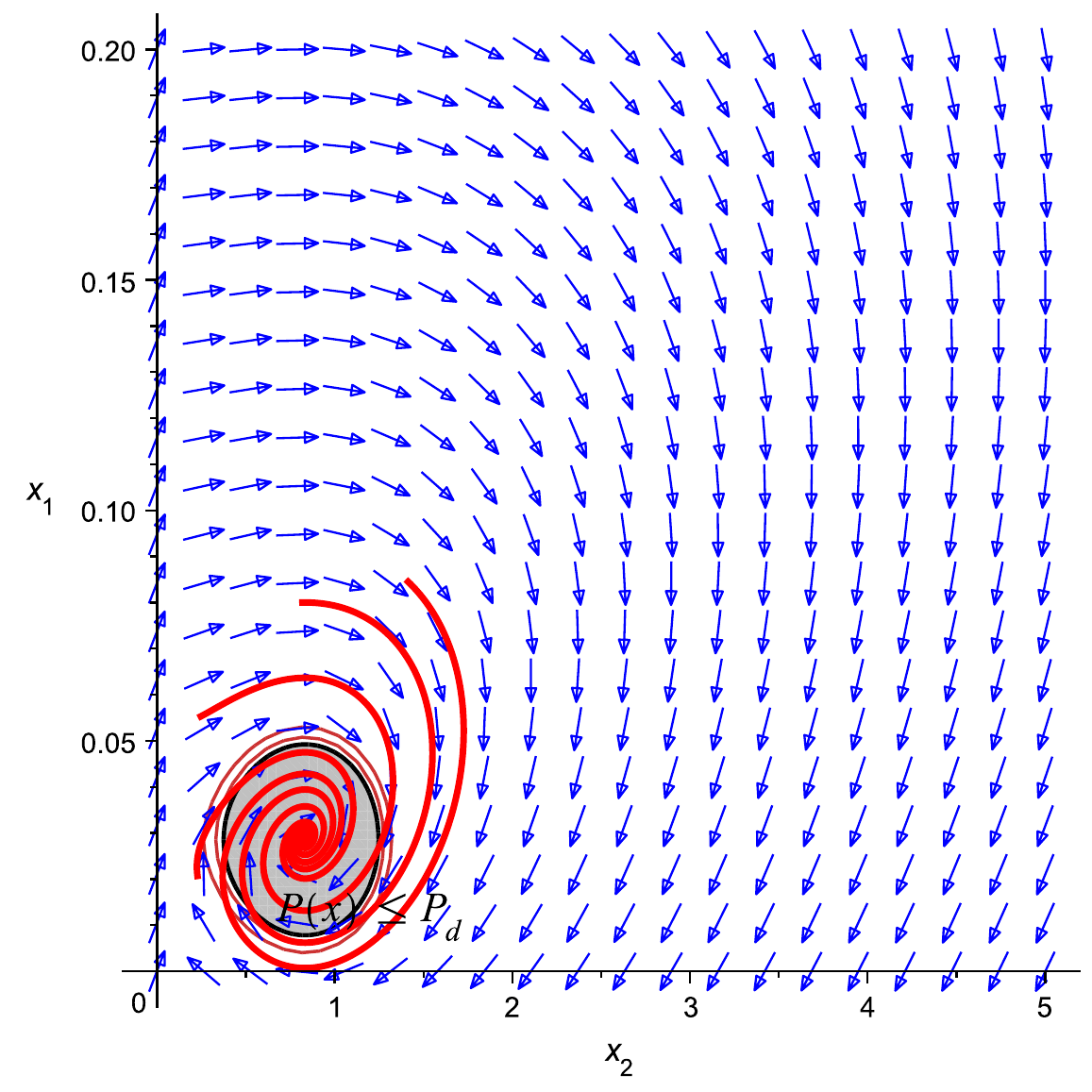}
    \\
    (a) The phase plot of the Buck converter
  \end{minipage}
  \hfill
  \begin{minipage}[b]{0.45\linewidth}
    \centering
    \includegraphics[width=\linewidth]{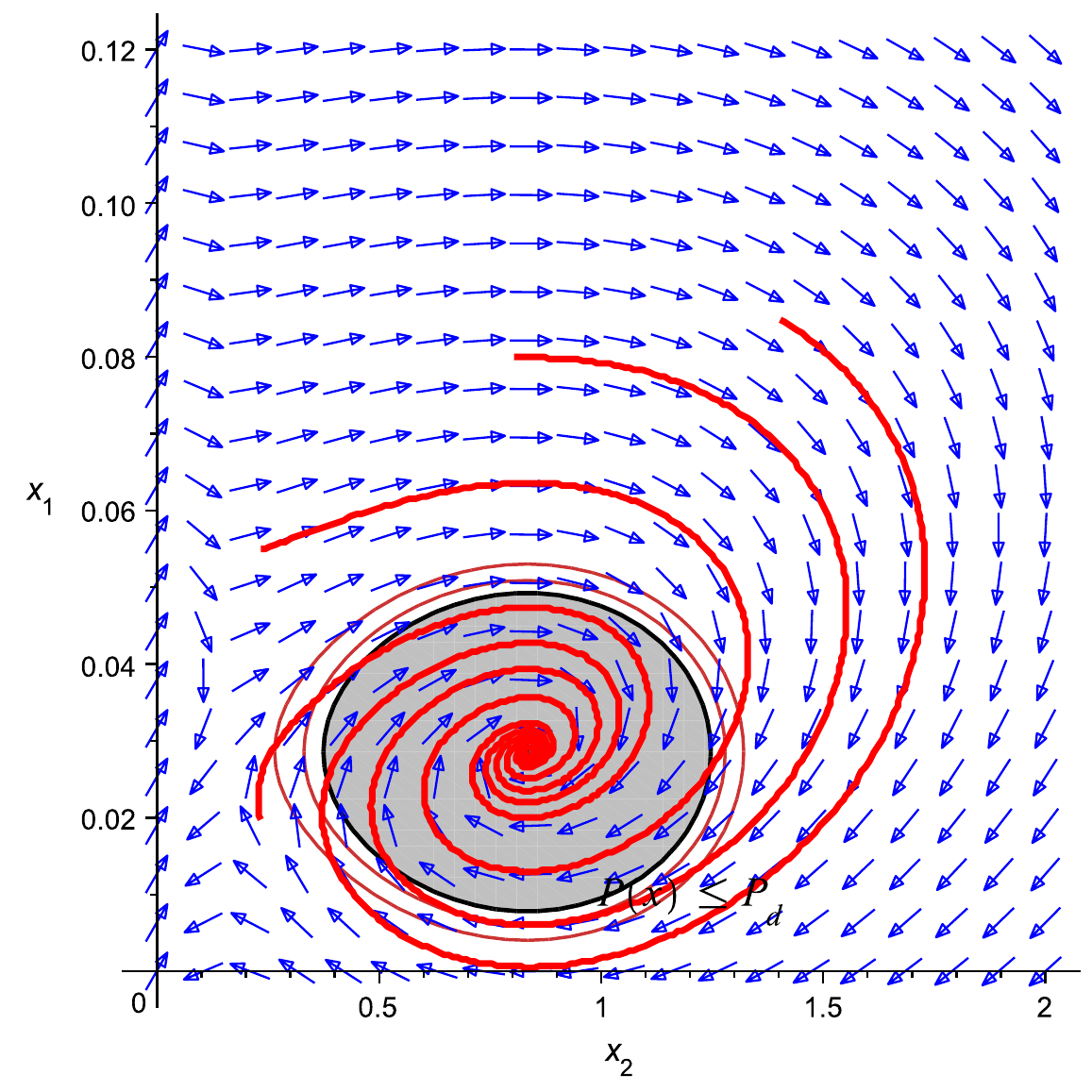}
    \\
    (b) Zoomed-in region of (a)
  \end{minipage}
  \caption{Phase plot of the Buck converter with zoomed-in view}  \label{bckf}
\end{figure}

%
\subsection{The Boost and Buck-Boost converters}
In this subsection we apply Proposition \ref{pro1} to design, {\em simultaneously}, the PBC for the Boost and the Buck-Boost converters.
\subsubsection{Proposed controller}

\begpro\label{pro3}\em
Consider the Boost or Buck-Boost converters dynamics given in \eqref{bst} and \eqref{gx2}, respectively. Fix $x_{2\star}$ verifying \eqref{ass1}. The mappings 
\begin{equation}
	\begin{aligned}\label{qbb}
		{\alpha}_1(x,u) &=-x_1,\;
		{\alpha}_2(x,u)  = 0 ,\;
		{\beta}(x,u) =-g(x_2)+ {k \over u},
	\end{aligned}
\end{equation}
and the control signal
\begequ
\label{ubst}
\hat u(x_2)= k\frac{h(x_2)}{h(x_2)g(x_2)+ c},
\endequ
where
$$
c:=(k-1)h_\star g_\star
$$
and the free constant $k$ is selected such that 
{ \begin{equation}\label{ubstcon}
	k\geq 1 +\frac{h_\star}{h'_\star g_\star},
\end{equation}}
fulfill the conditions \textbf{C1}-\textbf{C6} of Proposition \ref{pro1}. Consequently, both  systems in closed-loop with the control $u=\hat u(x_2)$ have an equilibrium point $(x_{1\star},x_{2\star})\in\mathcal{E}_{\mathrm{bb}}$ which is asymptotically stable.

Moreover, if the load is of the form \eqref{ym} a Lyapunov function for the equilibrium is given by: 

\begin{itemize}
	\item For the Boost Converter: 
	\begin{align*}
		P(x) =& \frac{1}{2} (k-1) \left(x_1-x_{1\star} \right)^2 +\frac{c}{2\RR}\ln\left[\frac{(x_2h(x_2) + c)^k}{ x_2 h(x_2)}\right]  
+{\frac{c}{2\RR} \ln \left(\frac{(kx_{1\star})^k}{x_{1\star}}\right)},
	\end{align*}
	\item For the  Buck-Boost Converter:
	\begin{align*}
		P(x)=& \frac{1}{2} (k-1) \left(x_1-x_{1\star} \right)^2+ \frac{k}{2}x_2^2 
 + k x_2 -\frac{c}{2\RR}\ln\left(\frac{x_2h}{x_{2\star} h_\star}\right) 
 - k \int_{x_{2}(0)}^{x_2(t)} \frac{g^2(s) h(s)}{g(s)h(s)+c} ds.
	\end{align*}  
\end{itemize}
\endpro
\begin{proof}
	The proof is given in Appendix \ref{appc}.
\end{proof}

\subsubsection{Numerical estimation of the region of attraction: Buck-Boost converter}
\lab{subsubsecbbc}

Consider the system \eqref{bst} in closed loop with \eqref{ubst}. Consider the circuit parameters given in Table \ref{para} and set the controller gain $k=3$. The desired equilibrium point is $(x_{1\star},x_{2\star})=(0.0881,1.25)$---which, according to \eqref{transf}, corresponds to an output voltage setpoint of $v_\star=30$V. {From straightforward computations, it can be seen that this value verifies \textbf{Assumption \ref{ass1}}  and \textbf{Remark \ref{condx2}}.}

We will show simulated results of the closed-loop system in the phase plane. {Fig. \ref{bbf} shows a set of trajectories, together with some level sets of the Lyapunov function $P(x)$ with {$\bar P\in \{0.0025,~0.004,~  0.005\}$.} Fig. \ref{bbf} (b) shows an enlarged view of Fig. \ref{bbf} (a)---in gray---an approximation of a subset of the domain of attraction that guarantees the whole trajectory remains in the positive (physically meaning) orthant. }

\begin{figure}[!ht]
  \centering
  \begin{minipage}[b]{0.45\linewidth}
    \centering
    \includegraphics[width=\linewidth]{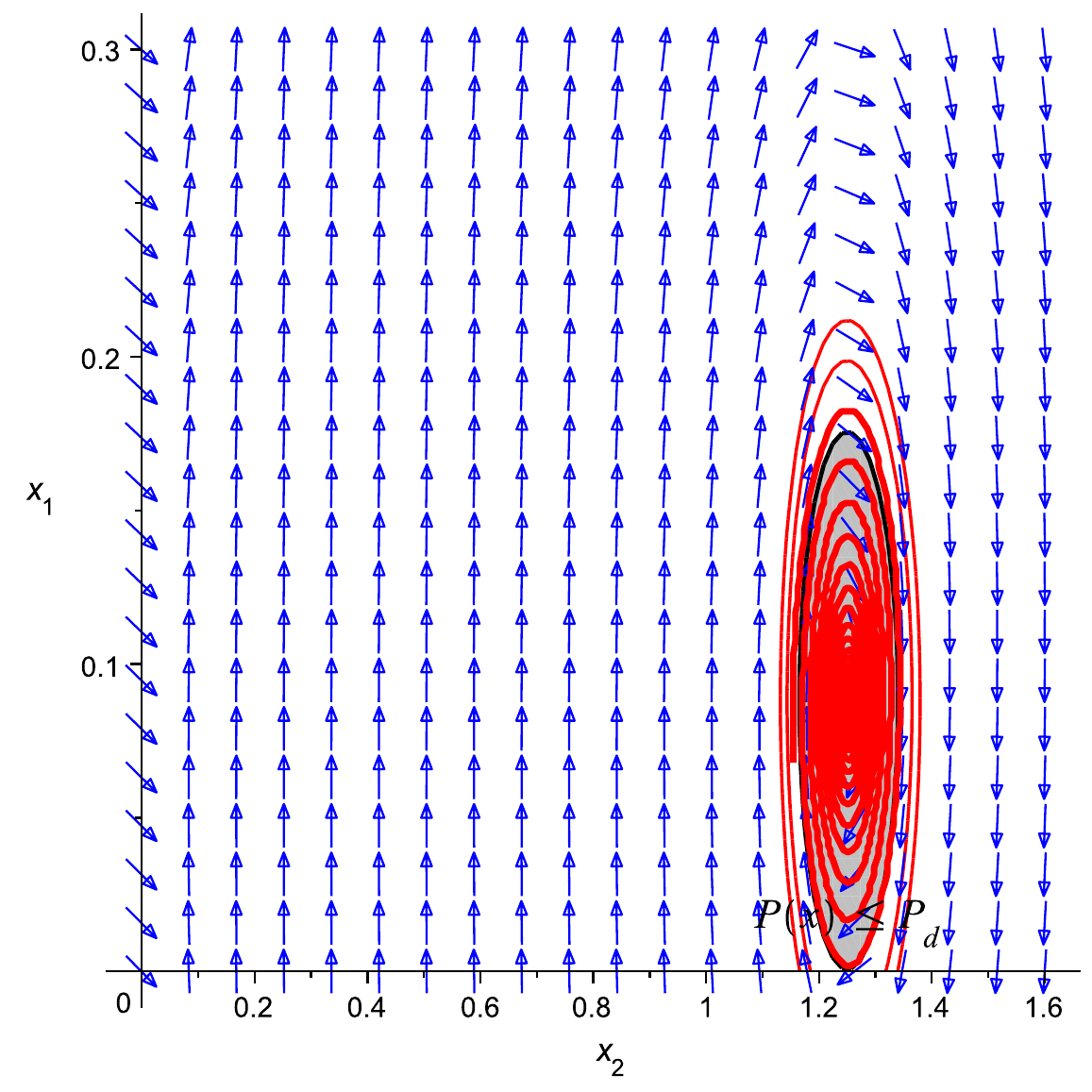}
    \\
    (a) The phase plot of the Buck-Boost converter
  \end{minipage}
  \hfill
  \begin{minipage}[b]{0.45\linewidth}
    \centering
    \includegraphics[width=\linewidth]{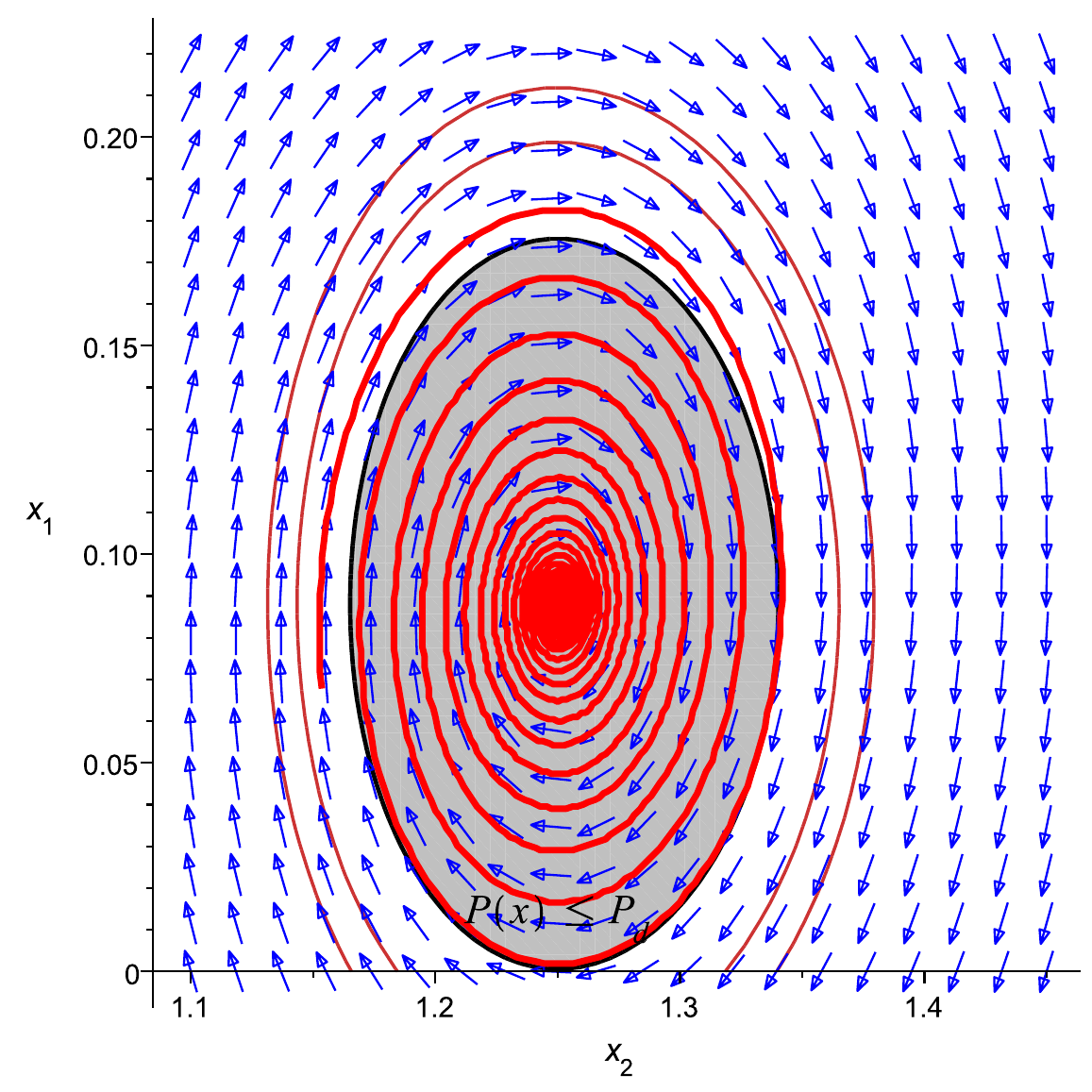}
    \\
    (b) Zoomed-in region of (a)
  \end{minipage}
  \caption{Phase plot of the Buck-Boost converter with zoomed-in view}  \label{bbf}
\end{figure}
%
\section{Indirect Adaptive Control for Load \eqref{ym}}
\lab{sec5}
%
We consider henceforth a load consisting of a CPL in parallel with a resistor as shown in Fig. \ref{ld}.  When the parameters $G$ and $P_{\mathrm{cpl}}$ are known,  $\PP$ and $\RR$ can be obtained from \eqref{param}.  To obtain the corresponding control laws for each converter, it only suffices to replace \eqref{hx2} into \eqref{ubuc} or  \eqref{ubst}.  In this section we assume that these parameters are {\em unknown} and present an identification algorithm that will generate an estimate of $G$ and $P_{\mathrm{cpl}}$, that will converge in {\em finite time} to their true value. This parameter estimates are used, on-line, in the controller yielding an indirect adaptive control scheme.

To that purpose, consider the relation between the load current $i_L$ {and $h(x_2)$ given in \eqref{hx2} that, taking into account the scaling factors \eqref{transf} and \eqref{hx2} yields the linear regression equation (LRE) 
\begin{align}\label{lre2}
	i_L(t)& = {E\sqrt{\frac{C}{L}} h(x_2(t))} =:\phi^\top(t) \theta,
\end{align}
where the vector signal $\phi(t) \in\rea^2$ and the parameter vector $\theta\in\rea^2$ are given by
$$
\phi := \begmat{ x_2 \\ \frac{1}{x_2} }, \; {\theta= \begmat{ \theta_1 \\ \theta_2 }= E\sqrt{\frac{C}{L}}\begmat{ \RR \\ \PP }=\begmat{ GE \\ \frac{P_{\tt cpl}}{E}}}.
$$

	\begin{remark}\em
	\lab{rem7}
		As mentioned above, the unknown parameters in an experimental setup are typically those of the load, i.e., $P_{\tt cpl}$ and $G$. From \eqref{param}, these constants can be recovered from $\theta$ as follows
		$$G= \frac{\theta_1}{E},\;\; P_{\tt cpl}=E\theta_2.$$
	\end{remark}

To implement the parameter estimator using the LRE \eqref{lre2} it is clear that we need to impose the following.

\begin{assumption}\em\label{as2}
	The load current  $i_L$ is {\em measurable}.
\end{assumption}

We are in position to state the main result of this section, where we propose to estimate the parameters $\theta$ using the standard least squares algorithm with FCT reported in \cite{ORTetalAUT_25}.

\begin{proposition}\em
	Consider the LRE \eqref{lre2} and Assumption \ref{as2}. Assume the vector $\phi(t)$ is IE \cite{KRERIE,TAObook} and bounded. That is, there exists $T_c>0$ and $\kappa>0$ such that
	\begequ
	\lab{ie}
	\int_0^{T_c} \phi(s)\phi^\top(s)ds \geq \kappa I_2.
	\endequ
	Define the standard LS estimator factor with forgetting factor
	\begin{subequations} \label{adap}
		\begin{align}
			\dot{\hat \theta }&= \gamma F \phi [i_L-\phi^\top \hat\theta], \;\hat\theta(0)=\theta_0\in\mathbb{R}^2,\\
			\dot F&=  -\gamma F \phi \phi^\top F+\chi F,\;F(0)=\frac{1}{f_0}I_2\\
			\dot z&=  -\chi z,\;z(0)=1,\\
			\chi&=  \chi_0\left(1-\frac{||F||}{\sigma}\right),
		\end{align}
	\end{subequations}
	with tuning gains the scalars $\gamma>0$, $f_0>0$, $\chi_0>0$ and $\sigma\geq\frac{1}{f_0}$.  For $t\geq T_c$, define the signal
	\begin{align*}
		\theta^{FCT}:&=  [I_2-zf_0F]^{-1}[\hat\theta-zf_0F\theta_0].
	\end{align*}
	The following statements are valid.
	\begenu[{\bf S1}]
	\item For all initial conditions this signal satisfies 
	$$
	\theta^{FCT}(t)=\theta,\;\forall t\geq T_c.
	$$
	\item All the signals are bounded.
	\item Assume the load of the converters is of the form \eqref{ym} with $h(x_2)$ verifying {\bf Assumption} \ref{ass1} for the desired voltage $x_{2\star}$. Consider the controllers $\hat u(x_2)$ given in Propositions \ref{pro2} and \ref{pro3} with $h(x_2)$ replaced by
	$$
	\hat h(x_2)= \theta_1^{FCT} x_2 + \frac{\theta_2^{FCT} }{x_2}.
	$$
	The associated equilibrium points $(x_{1\star},x_{2\star})$ are asymptotically stable.
	\endenu
\end{proposition}

\begin{proof}
	The proof of {\bf S1} and {\bf S2} is given in  \cite{ORTetalAUT_25}. The claim {\bf S3} follows trivially from the FCT global convergence of the parameter estimates and the global stability proofs of Propositions \ref{pro2} and \ref{pro3}. 
\end{proof}

%
\section{Simulation Results}
\lab{sec6}
%
{This section illustrates the performance of the adaptive IDA-PBC for the normalized models of the three converters, with the load represented by $i_L$ (see Eq. \eqref{lre2}) and the physical parameters listed in Table \ref{para}. The unknown parameters of the normalized models are
$$
     \begmat{ \theta_1 \\ \theta_2 }=\begmat{ GE \\ \frac{P_{\tt cpl}}{E} }= \begmat{ 0.4\\ 0.05 }.
$$
}

\subsection{Buck converter} 

The IDA-PBC method is applied with control law \eqref{ubuc} and adaptive law \eqref{adap}, using design parameters $k=0.1$, $\gamma=10$, $\chi_0=1$, $\sigma=10$, $f_0=4$ and the initial conditions $[x_1(0),x_2(0)]^\top=[0.015,1.15]^\top$, $[\hat\theta_1(0),\hat\theta_2(0)]^\top=[0.01,0.002]^\top$.

{The desired equilibrium point is that given in Section \ref{subsubsecbc}}. The simulation results are shown in Fig. \ref{simfig1} and Fig. \ref{simfig2}. In Fig. 7(a)-7(c), the evolution of $x_1$,  $x_2$, and  $u$ as time increases is shown. Notice their convergence to the corresponding equilibrium values, validating the effectiveness of the proposed control method. On the other hand, Fig. \ref{simfig2} shows that the estimated parameters $\hat\theta_1$ and $\hat\theta_2$ quickly converge to their true values, demonstrating the accuracy and reliability of the adaptive laws.}

\begin{figure}[!h]
	\centering
	\begin{minipage}[b]{0.4\linewidth}
		\centering
		\includegraphics[width=\linewidth]{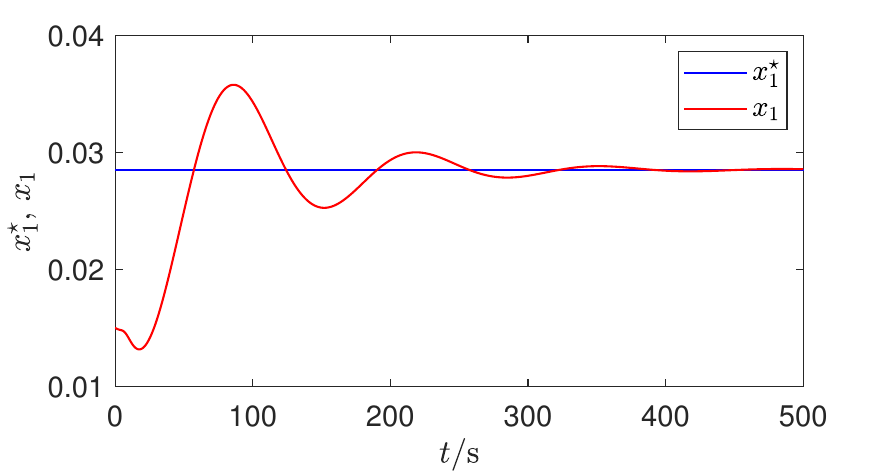}
		\\
		(a) The state $x_1$ and its desired value
	\end{minipage}
	\hfill
	\begin{minipage}[b]{0.4\linewidth}
		\centering
		\includegraphics[width=\linewidth]{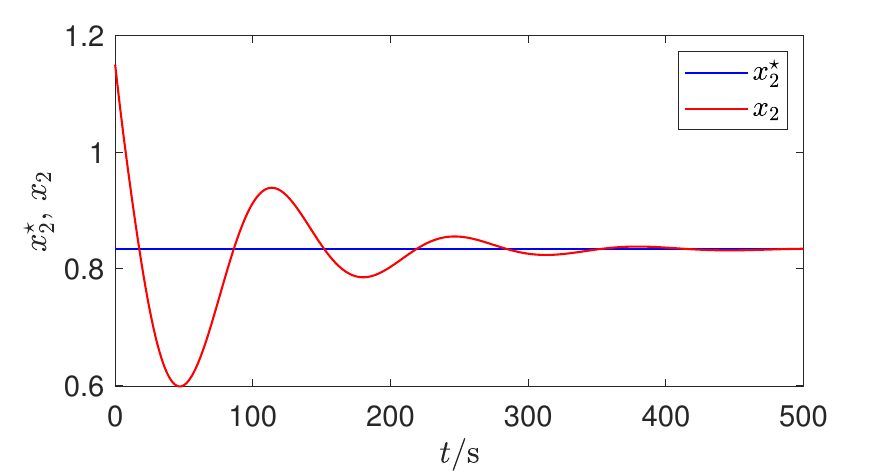}
		\\
		(b) The state $x_2$ and its desired value
	\end{minipage}
	\hfill
	\begin{minipage}[b]{0.4\linewidth}
		\centering
		\includegraphics[width=\linewidth]{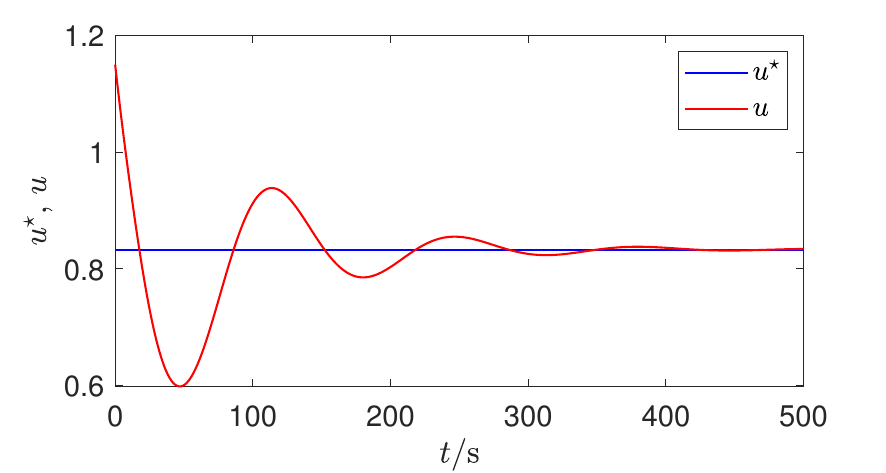}
		\\
		(c) The control input $u$ and its desired value
	\end{minipage}
	\caption{The states and control input of the Buck converter. }
	\label{simfig1}
\end{figure}

\begin{figure}[!h]
  \centering
  \begin{minipage}[b]{0.4\linewidth}
    \centering
    \includegraphics[width=\linewidth]{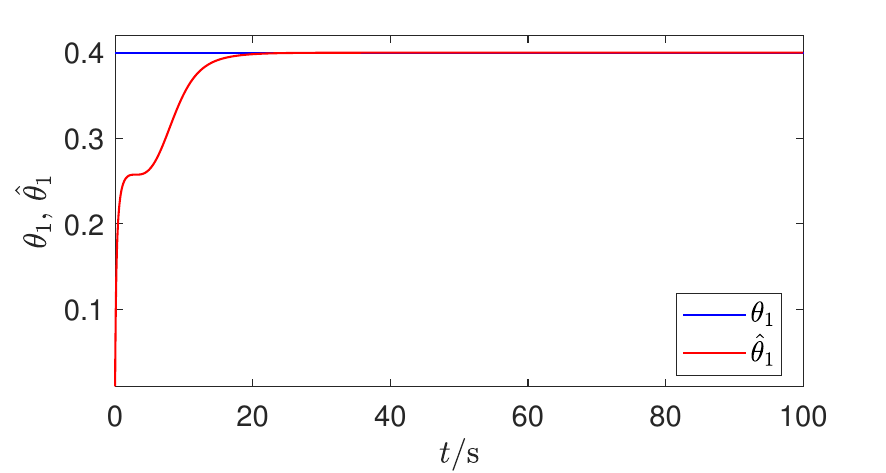}
    \\
    (a) The estimated value $\hat{\theta}_1$ and its desired value
  \end{minipage}
  \hfill
  \begin{minipage}[b]{0.4\linewidth}
    \centering
    \includegraphics[width=\linewidth]{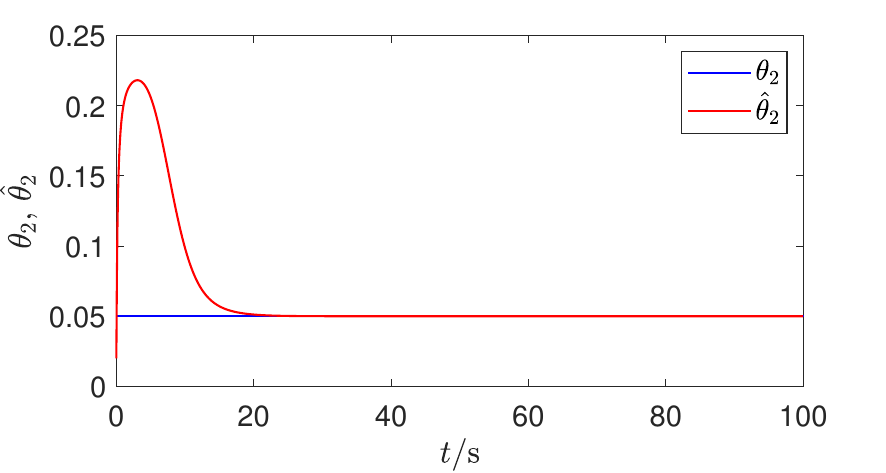}
    \\
    (b) The estimated value $\hat{\theta}_2$ and its desired value
  \end{minipage}
  \caption{The estimated values of the Buck converter.}
  \label{simfig2}
\end{figure}

\subsection{Buck-Boost converter} 
The IDA-PBC method is applied with control law \eqref{ubst} and adaptive law \eqref{adap}, using design parameters $k=1.6523$, $\gamma=15$, $\chi_0=1$, $\sigma=10$, $f_0=4$ and the initial conditions $[x_1(0),x_2(0)]^\top=[0.1,1.3]^\top$, $[\hat\theta_1(0),\hat\theta_2(0)]^\top=[0.01,0.002]^\top$.

{The desired equilibrium point is that given in Section \ref{subsubsecbbc}}. Simulation results are shown in Fig. \ref{simfig3} to Fig. \ref{simfig4}. First, Fig. 9(a)-9(c) show the time response of $x_1$, $x_2$ and $u$, all of which converge to their desired values with small-amplitude but relatively high-frequency oscillations, validating the effectiveness of the proposed control method. Fig. \ref{simfig4} shows that that the proposed adaptive laws enable the estimated parameters $\hat\theta_1$ and $\hat\theta_2$ to accurately track their true values in a short time.

\begin{figure}[!h]
  \centering
  \begin{minipage}[b]{0.4\linewidth}
    \centering
    \includegraphics[width=\linewidth]{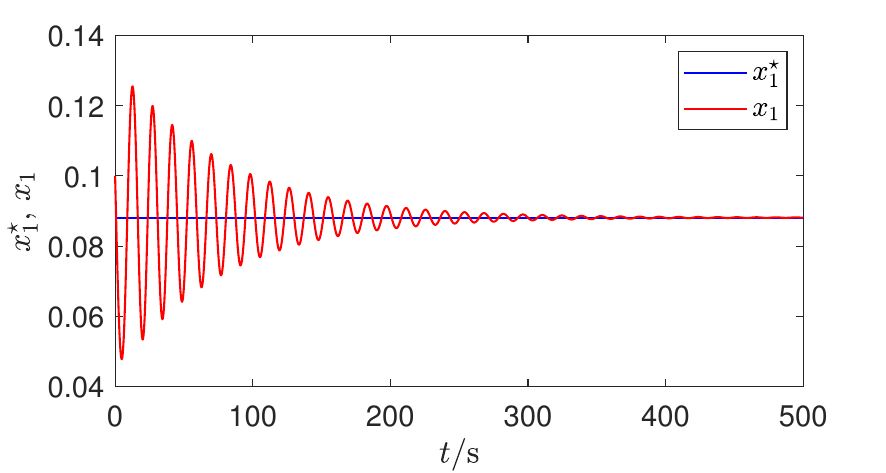}
    \\
    (a) The state $x_1$ and its desired value
  \end{minipage}
  \hfill
  \begin{minipage}[b]{0.4\linewidth}
    \centering
    \includegraphics[width=\linewidth]{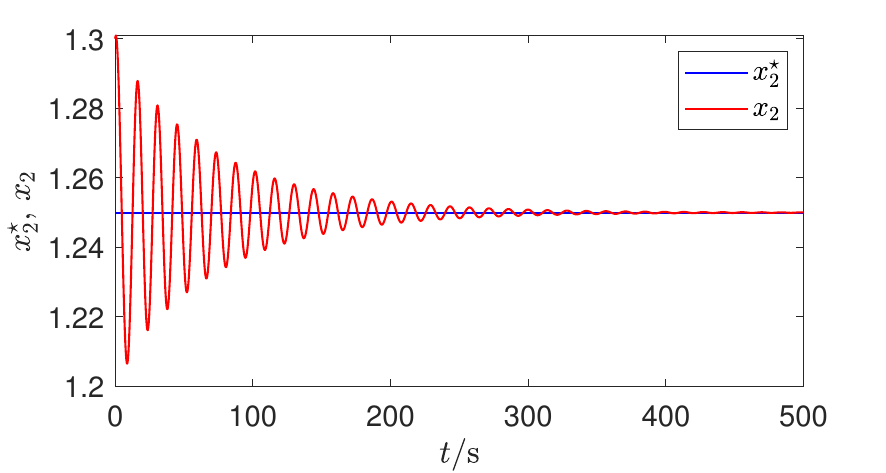}
    \\
    (b) The state $x_2$ and its desired value
  \end{minipage}
  \hfill
  \begin{minipage}[b]{0.4\linewidth}
    \centering
    \includegraphics[width=\linewidth]{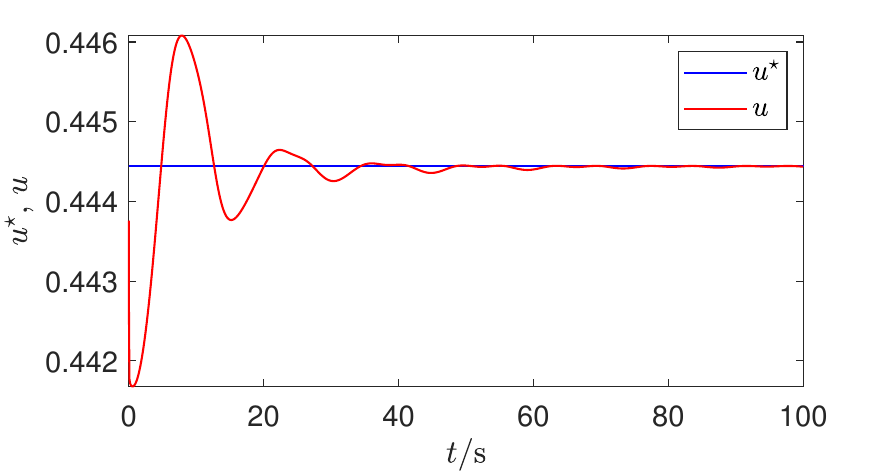}
    \\
    (c) The control input $u$ and its desired value
  \end{minipage}
  \caption{The states and control input of the Buck-Boost converter.}
  \label{simfig3}
\end{figure}

\begin{figure}[!ht]
  \centering
  \begin{minipage}[b]{0.4\linewidth}
    \centering
    \includegraphics[width=\linewidth]{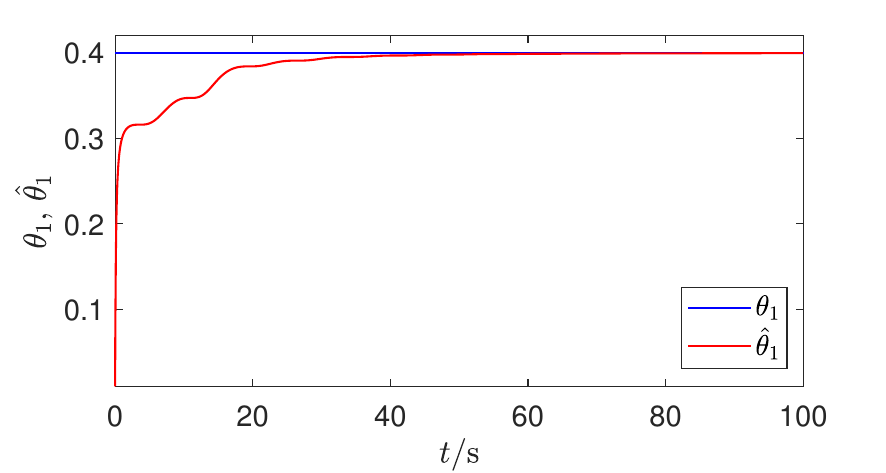}
    \\
    (a) The estimated value $\hat{\theta}_1$ and its desired value
  \end{minipage}
  \hfill
  \begin{minipage}[b]{0.4\linewidth}
    \centering
    \includegraphics[width=\linewidth]{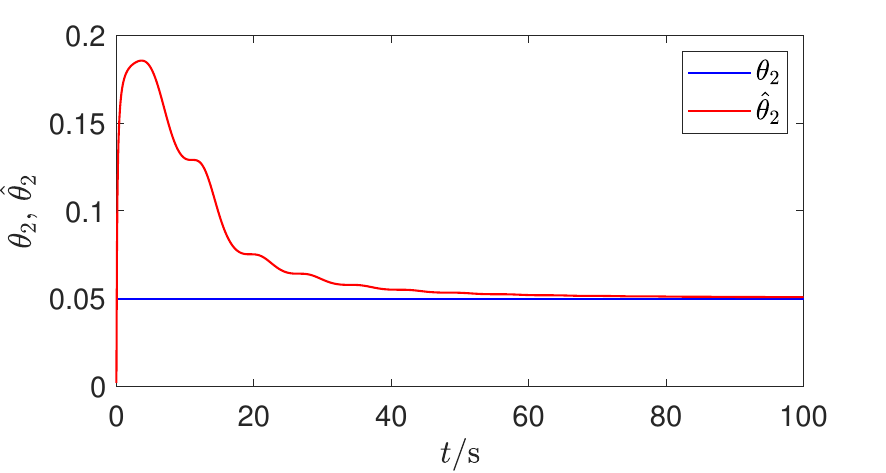}
    \\
    (b) The estimated value $\hat{\theta}_2$ and its desired value
  \end{minipage}
  \caption{The estimated values of the Buck-Boost converter.}
  \label{simfig4}
\end{figure}

\section{Experimental results}
\label{sec7}

{In this section, we assume that {the parameters to be estimated are known} and show experimental results of the proposed control method for three converters. The experimental setup incorporates a combination of $R$ and $P_{\tt cpl}$, as shown in Fig. \ref{setup}. The experimental procedure is as follows: the control algorithm is compiled into a C program and executed on the YXSPACE controller equipped with a TMS320F28335 microcontroller; the controller regulates the output voltage of the converters in real time through its input and output ports to achieve the desired values. The nominal circuit parameters of the three converters are listed in Table \ref{para}.

\begin{figure}
	\centering
	\includegraphics[scale=1]{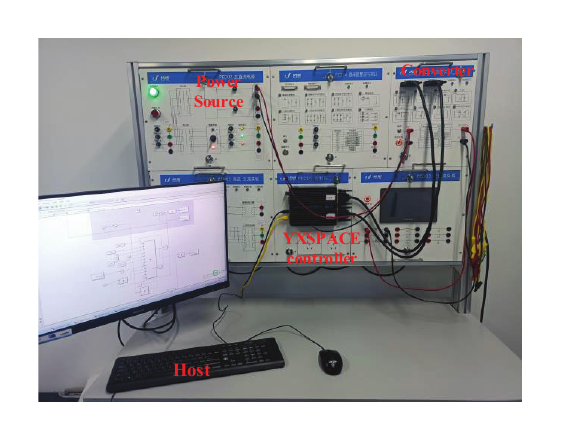}
	\caption{The setup.}
	\label{setup}
\end{figure}

\subsection{Buck converter} 

The reference voltage $v_\star$ is initially set to $20$ V, then decreased to $15$ V, and further reduced to $10$ V. The control gains are chosen as $k=0.01$ and $\gamma=10$. As shown in Fig. \ref{expfig1} (a), the output voltage temporarily reaches the saturation limit of $22$ V, but quickly converges to the initial reference voltage of $20$ V and subsequently follows the reference voltage as it decreases to $15$ V and $10$ V, demonstrating good stability of the proposed control method. The robustness of the proposed method under load variations is further validated. The reference voltage is fixed at $15$ V, while $R$ changes from $60~\Omega$ to $30~\Omega$ and $P_{\tt cpl}$ varies from $1.2$ W to $1.8$ W. The experimental results, shown in Fig. \ref{expfig2} (a), indicate that the output voltage is barely affected by the changes in $R$ and $P_{\tt cpl}$ and accurately converge to the reference value, demonstrating the strong robustness of the proposed control method against load variations. Figures \ref{expfig1} (b) and \ref{expfig2} (b) respectively show the control inputs under these two conditions.

\begin{figure}[!ht]
  \centering
  \begin{minipage}[b]{0.35\linewidth}
    \centering
    \includegraphics[width=\linewidth]{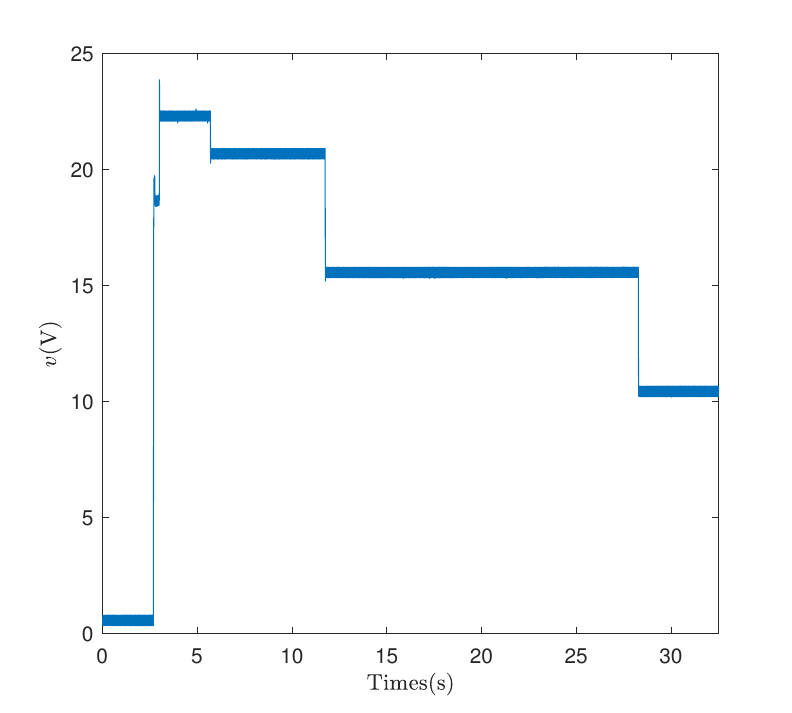}
    \\
    (a) output voltage
  \end{minipage}
  \hfil
  \begin{minipage}[b]{0.35\linewidth}
    \centering
    \includegraphics[width=\linewidth]{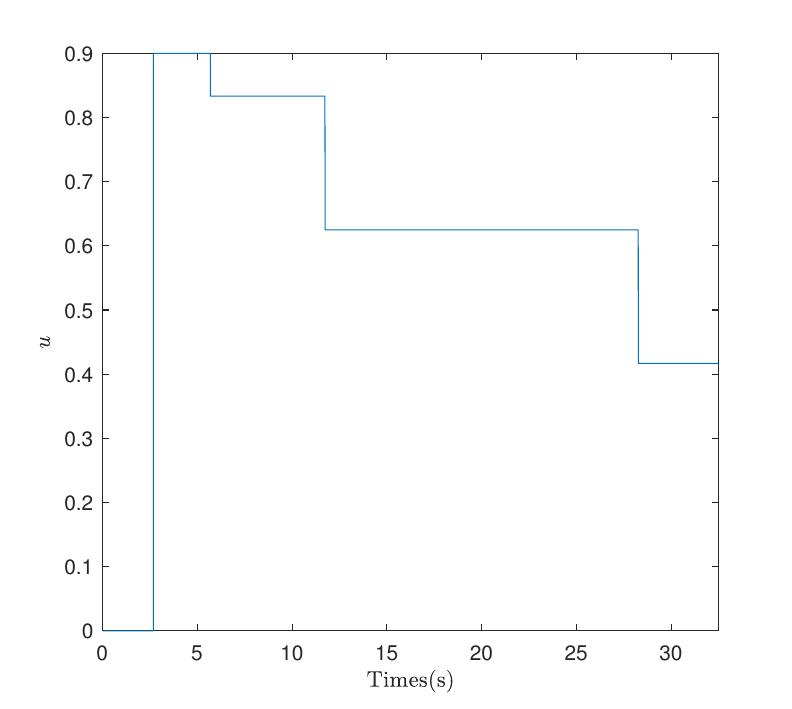}
    \\
    (b) duty ratio
  \end{minipage}
  \caption{Response curves of Buck converter under step change in $v_\star$}
  \label{expfig1}
\end{figure}

\begin{figure}[!ht]
  \centering
  \begin{minipage}[b]{0.35\linewidth}
    \centering
    \includegraphics[width=\linewidth]{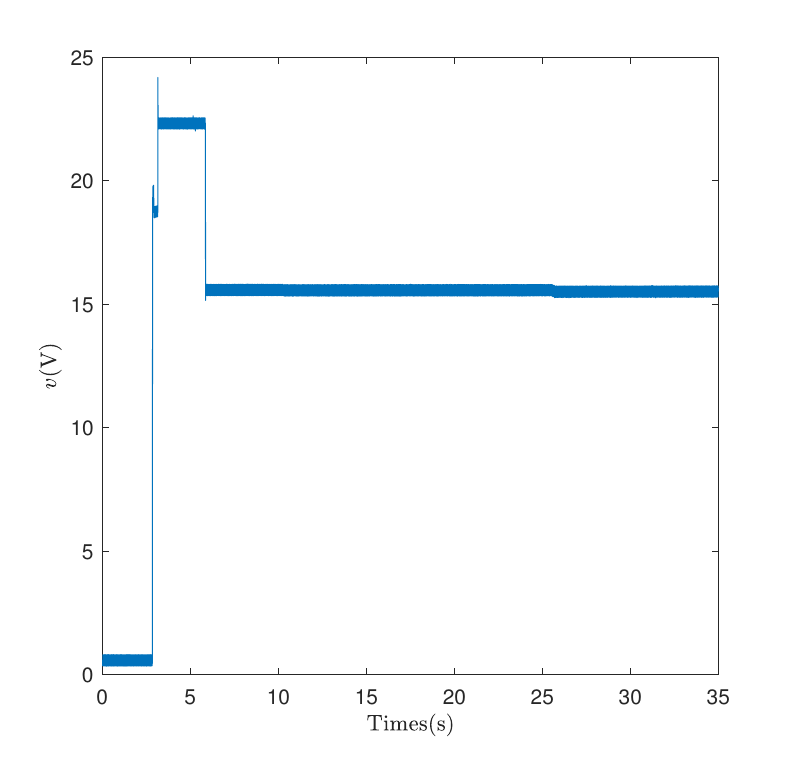}
    \\
    (a) output voltage
  \end{minipage}
  \hfil
  \begin{minipage}[b]{0.35\linewidth}
    \centering
    \includegraphics[width=\linewidth]{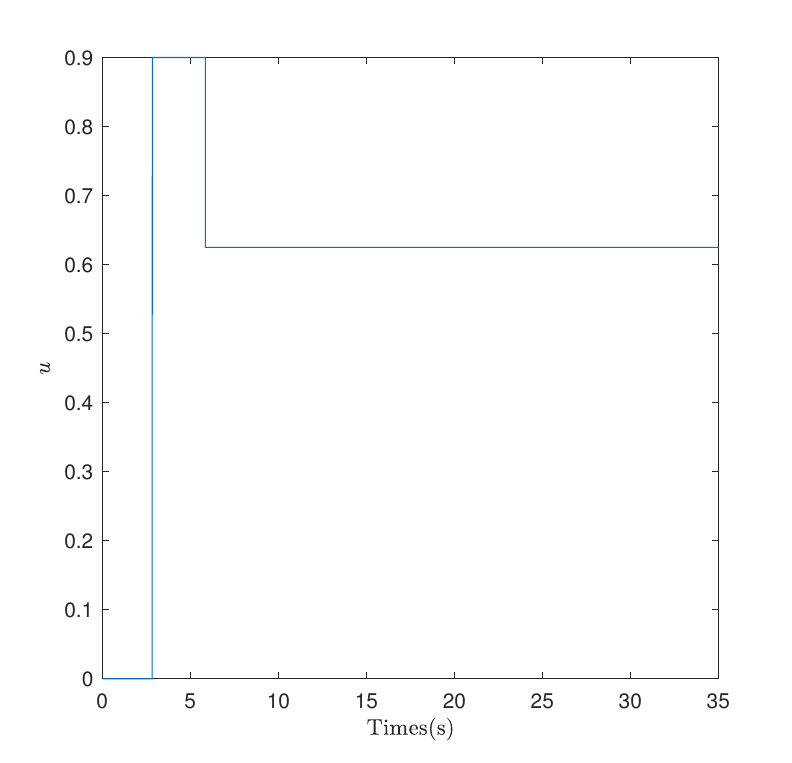}
    \\
    (b) duty ratio 
  \end{minipage}
  \caption{Response curves of Buck converter under step change in $R$ and $P_{\tt cpl}$}
  \label{expfig2}
\end{figure}

\subsection{Boost converter} 
The control gains were set to  $k=3$ and $\gamma=10$. The reference voltage $v_\star$
 is initially set to $26$ V, and then sequentially stepped to $30$ V and $40$ V. As shown in Fig. \ref{expfig3} (a), the output voltage exhibits brief deviations in response to these steps but quickly stabilizes near the reference values, with only minor steady-state errors due to unmodeled dynamics, demonstrating the good stability of the proposed control method. The robustness of the proposed method under load variations is further evaluated. The reference voltage is fixed at $30$ V, while $R$ changes from $60~\Omega$ to $30~\Omega$ and $P_{\tt cpl}$ varies from $1.2$ W to $1.8$ W. As shown in Fig. \ref{expfig4} (a), the output voltage exhibits only minor fluctuations in response to the load changes and eventually stabilizes near the reference value, with only small steady-state errors which may be caused by parasitic resistances. These results indicate that the proposed control method maintains good robustness under load variations.  Figures \ref{expfig3} (b) and \ref{expfig4} (b) respectively show the control inputs under these two conditions.
 
\begin{figure}[!ht]
  \centering
  \begin{minipage}[b]{0.35\linewidth}
    \centering
    \includegraphics[width=\linewidth]{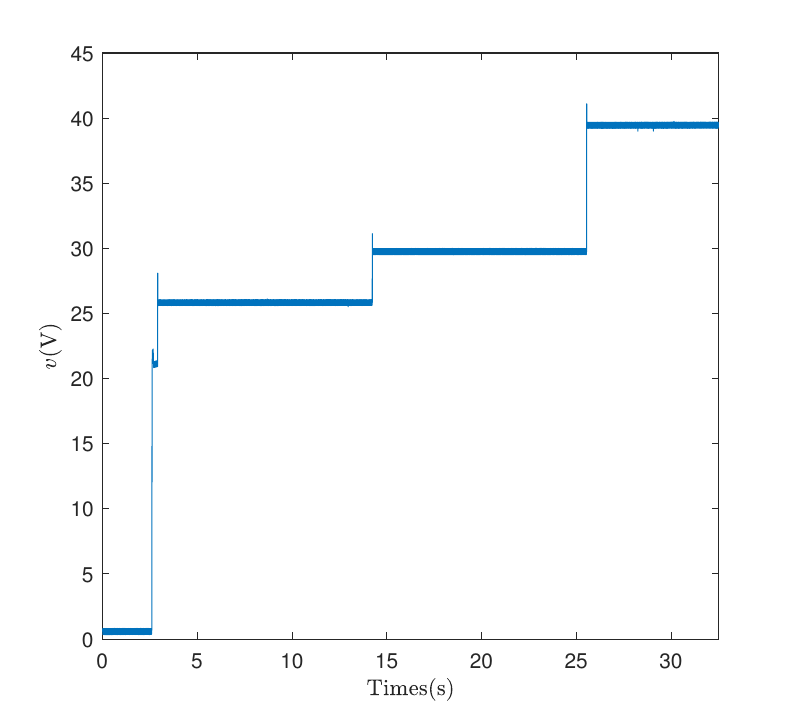}
    \\
    (a) output voltage
  \end{minipage}
  \hfil
  \begin{minipage}[b]{0.35\linewidth}
    \centering
    \includegraphics[width=\linewidth]{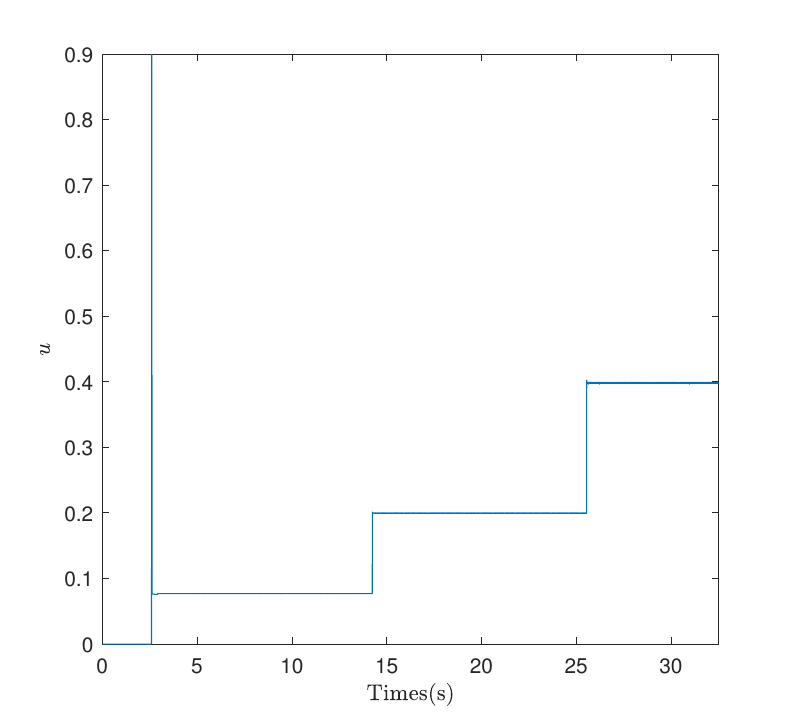}
    \\
    (b) duty ratio
  \end{minipage}
  \caption{Response curves of Boost converter under step change in $v_\star$}
  \label{expfig3}
\end{figure}

\begin{figure}[!ht]
  \centering
  \begin{minipage}[b]{0.35\linewidth}
    \centering
    \includegraphics[width=\linewidth]{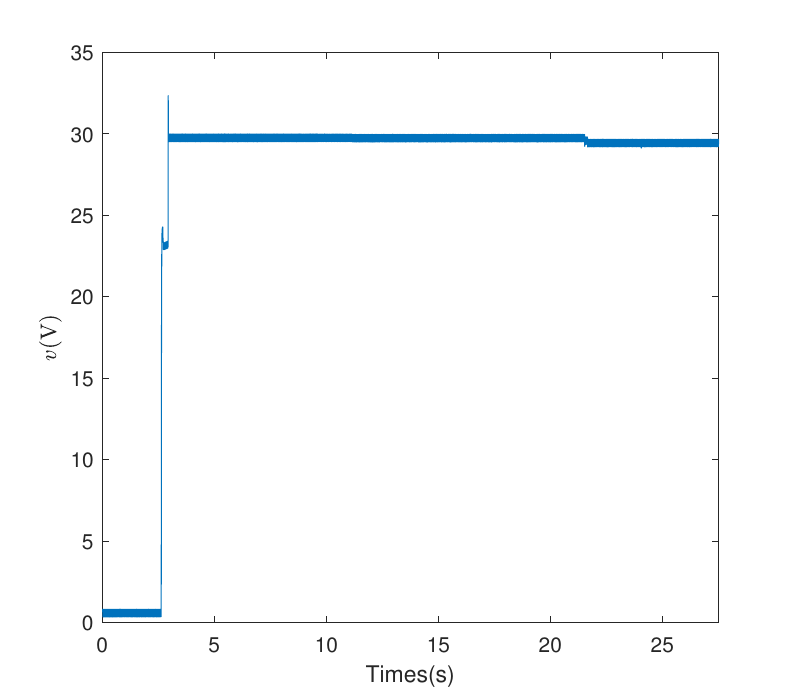}
    \\
    (a) output voltage
  \end{minipage}
  \hfil
  \begin{minipage}[b]{0.35\linewidth}
    \centering
    \includegraphics[width=\linewidth]{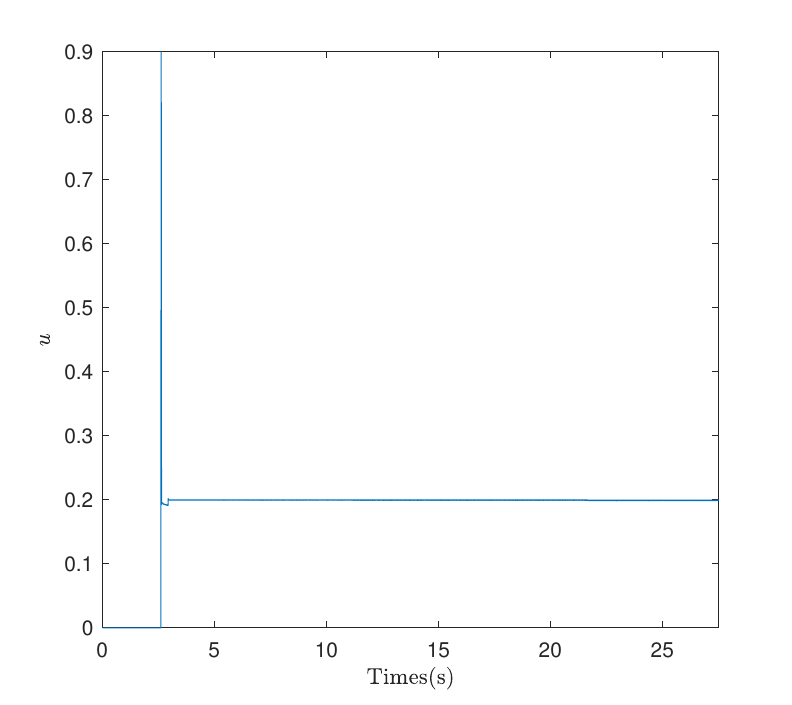}
    \\
    (b) duty ratio
  \end{minipage}
  \caption{Response curves of Boost converter under step change in $R$ and $P_{\tt cpl}$}
  \label{expfig4}
\end{figure}

\subsection{Buck-Boost converter}
The control gains were set to 
$k=2$ and $\gamma=10$. The reference voltage $v_\star$ is initially set to $20$ V, and then sequentially stepped to $24$ V and $30$ V. Fig. \ref{expfig5} (a) shows that the output voltage responded rapidly to these changes and stabilizes near the desired values, with minimal steady-state error primarily caused by unmodeled system dynamics. Subsequently, with the reference voltage fixed at $30$ V, $R$ is decreased from $60~\Omega$ to $30~\Omega$ and $P_{\tt cpl}$ is increased from $1.2$ W to $1.8$ W. Fig. \ref{expfig6}(a) shows that the output voltage converges close to the reference value, responding smoothly to the load variations and exhibiting only minor steady-state deviations, which may be caused by parasitic resistances. Figures \ref{expfig5} (b) and \ref{expfig6} (b) respectively show the control inputs under these two conditions. Overall, the proposed control method demonstrates excellent stability and robustness performance.}

\begin{figure}[!ht]
  \centering
  \begin{minipage}[b]{0.35\linewidth}
    \centering
    \includegraphics[width=\linewidth]{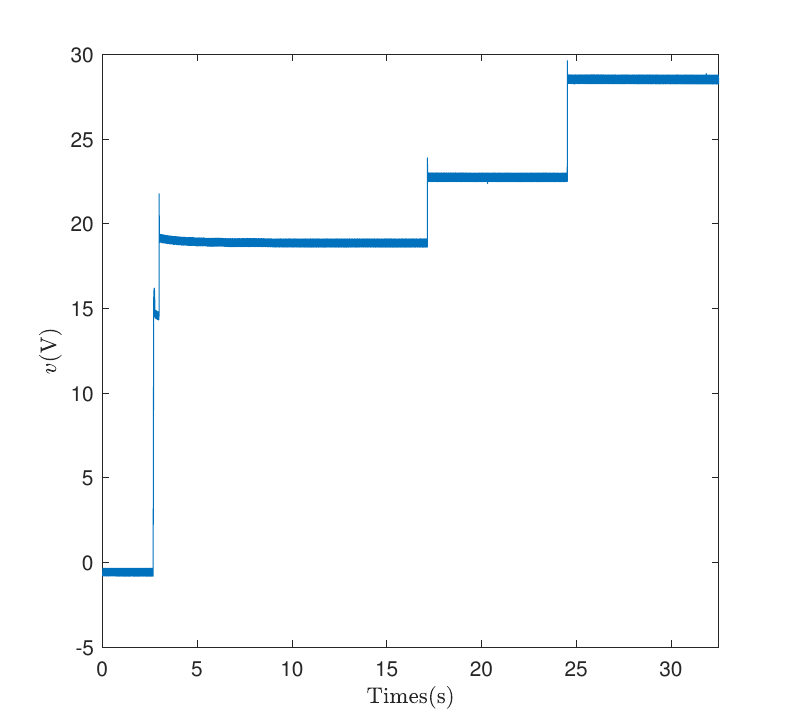}
    \\
    (a) output voltage
  \end{minipage}
  \hfil
  \begin{minipage}[b]{0.35\linewidth}
    \centering
    \includegraphics[width=\linewidth]{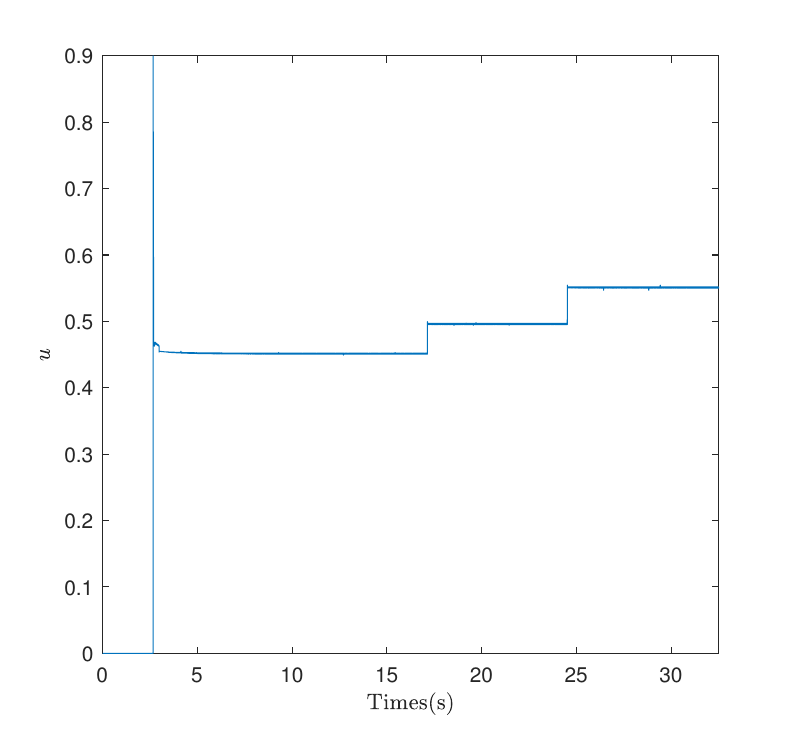}
    \\
    (b) duty ratio
  \end{minipage}
  \caption{Response curves of Buck-Boost converter under step change in $v_\star$}
  \label{expfig5}
\end{figure}

\begin{figure}[!ht]
  \centering
  \begin{minipage}[b]{0.35\linewidth}
    \centering
    \includegraphics[width=\linewidth]{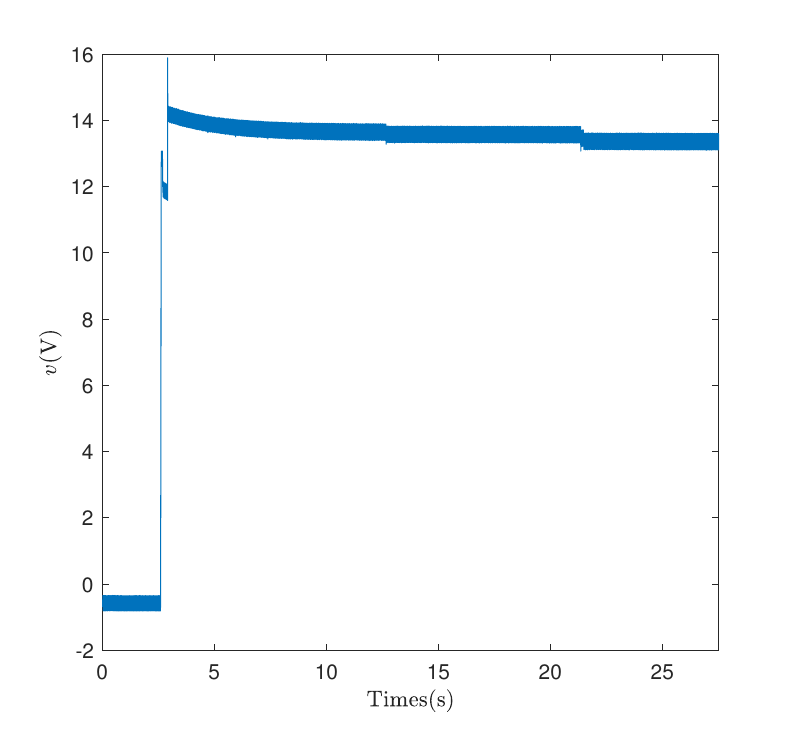}
    \\
    (a) output voltage
  \end{minipage}
  \hfil
  \begin{minipage}[b]{0.35\linewidth}
    \centering
    \includegraphics[width=\linewidth]{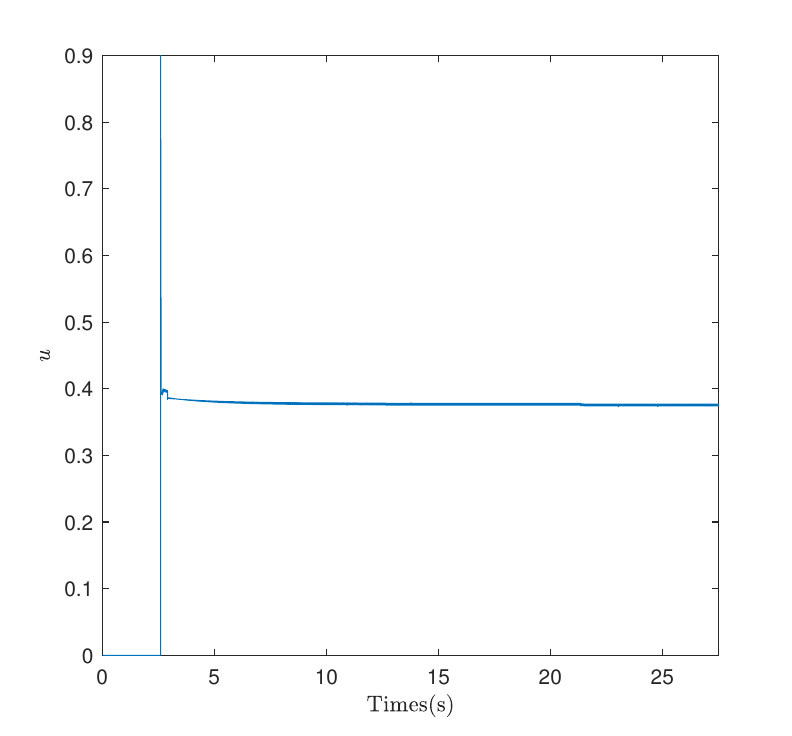}
    \\
    (b) duty ratio
  \end{minipage}
  \caption{Response curves of Buck-Boost converter under step change in $R$ and $P_{\tt cpl}$}
  \label{expfig6}
\end{figure}

%
\section{Concluding Remarks and Future Work}
\lab{sec8}
%
We have proposed a new procedure to design partial state feedback IDA-PBC for second order systems, which are possibly nonlinear in the control signal. One important feature of the procedure is that there is no need to solve the matching PDE, that often stymies the application of IDA-PBC. Instead, it is only required to solve an ODE, that in some cases, is trivially solved. 

The procedure is applied for the design of voltage regulation controllers for three widely popular DC-to-DC power converters, The main feature of these schemes is that   we assume only measurement of the output voltage that---due to its non-minimum phase characteristic---is well-known to be a difficult task. Furthermore, the controllers are applicable under a very general scenario for the load, namely, it is only assumed to be known a static relationship between its current and voltage. A very mild assumption on the behavior of this function at the equilibrium point is imposed.

For the, often encountered, case of a load consisting of the parallel connection of linear resistor and and a constant power load, we moreover propose a very simple parameter estimation scheme, that enables the application of an adaptive controller. In this particular case of a load, we also derive a Lyapunov function for the stabilized equilibrium, that allows us to estimate their domain of attraction. 

Finally, we provide simulation and experimental evidence of the proposed controllers---leaving for future research the experimental validation of the adaptive scheme and the comparison with other controllers reported in the literature---{\em e.g.}, the voltage feedback PI schemes reported in \cite{FANORTGRI}.

\bibliography{NewIDAPBC_v39_ExtendVersion}
\bibliographystyle{IEEEtranTIE}

\appendices
\section{Poincare's Lemma}
\lab{appa}

\begin{lemma} \em \cite[Theorem 10.39]{RUDbook}
	The map $g:\rea^n \to \rea^n$ satisfies
	$$
	g(x)=\nabla v(x),
	$$
	for some $v: \rea^n \to \rea$, if and only if
	$$
	\nabla g(x)=[\nabla g(x)]^\top.
	$$
\end{lemma}

\section{Proof of Proposition \ref{pro2}}
\lab{appb}
First, we make the observation that ${\alpha}_1$ and ${\beta}$ are {\em constants}, independent of $x$ and $u$. Therefore, $\hat {\alpha}_1={\alpha}_1$ and $\hat {\beta}={\beta}$. Clearly, these constants satisfy  $\textbf{C1}$ and   $\textbf{C2}$.\footnote{It is important to note that this is the {\em simplest} choice one can make for these functions, satisfying conditions $\textbf{C1}$ and   $\textbf{C2}$.} On the other hand, $D_1(x,u)$ and $D_2(x,u)$ are given by
	\begalis{
		D_1(x,u) &= \frac{1}{k} (x_2 - u)+ x_1-h(x_2) \\
		D_2(x,u) &=  x_2-u,
	}
	Computing the partial derivatives
	\begalis{
		\nabla_{x_2}D_1(x,\hat u(x_2)) &= \frac{1}{k}-h'(x_2) -\frac{1}{k} \hat u'(x_2) \\
		\nabla_{x_1}D_2(x,\hat u(x_2)) &= 0.
	}
	The equation \eqref{poilem} of \textbf{C3} corresponds, in this case, to the ODE
	\begin{align}\label{odebck}
		\hat u'(x_2) =1 -k h'(x_2), 
	\end{align}
	whose solution is trivially given by
	$$
	\hat u(x_2) =x_2 - k h(x_2)+ \kappa, 
	$$
	where $\kappa$ is a free constant. To satisfy \textbf{C4} we impose the constraint $u_\star=x_{2\star}$, yielding
	$$
	\kappa=k h_\star,
	$$
that, replaced in the equation above, corresponds to the control given in the proposition. 
	
	Replacing the control $\hat u(x_2)$ in the vector $D(x,u)$ we get
	\begequ
	\lab{e1}
	\hat D(x)=\begin{bmatrix}
		x_1-h_\star \\ \\ k[h(x_2)-h_\star]
	\end{bmatrix}
	\endequ
	Its Hessian is
	$$
	\nabla \hat D = \begin{bmatrix} 1& 0 \\ 0& kh'(x_2)\end{bmatrix}.
	$$
	Evaluating at $x_{2\star}$ and imposing {\bf Assumption} \ref{ass1} we verify the condition \eqref{codp} in {\bf C5}.
	
	Finally, some simple calculations show that the elements of the set \eqref{asysta} must satisfy $h(x_2)=h_\star$, whose only solution is $x_2=x_{2\star}$---verifying condition {\bf C6}.
	
	We proceed now to obtain the Lyapunov function  $P(x)$ defined in \eqref{Potbck} when the load is of the form \eqref{ym}. This is obtained noting that {\bf C3} ensures the existence of a scalar function $P(x)$ such that $\hat D(x)=\nabla P(x)$. Hence, we integrate $\hat D(x)$ given in \eqref{e1}.
	Integrating the first element by $x_1$ yields 
	\begin{align*}
		P(x)= \frac{x_1^2}{2}-h_\star x_1 +\mu(x_2) 
	\end{align*}
	where the function $\mu (x_2)$ is an ``integration constant". This implies that 
	$$
	\frac{\partial P}{\partial x_2}= \mu '(x_2) \equiv \hat D_2(x_2)=k[h(x_2)-h_\star].
	$$ 
	Therefore, $\mu(x_2)$ is obtained from the integration of the right side element of the last equality. An adequate selection of the integration constant,  such that $P(x_\star)=0$, yields the expression in \eqref{Potbck}. 

\section{Proof of Proposition \ref{pro3}}
\lab{appc}
	The polynomial of $\textbf{C1}$ becomes 
	\begin{align*}
		\hat\alpha_1(x)\hat\alpha_2(x)+\hat\beta^2(x)=&\left(g(x_2)-\frac{k}{\hat u(x_2)}\right)^2 =\frac{c^2}{h^2(x_2)}  \neq 0,
	\end{align*}
	where \eqref{ubst} was substituted to obtain the last expression. Therefore, \textbf{C1} is met. Furthermore, for the physically constrained variables, it is clear that \textbf{C2} is fulfilled as well. 
	
	Now, the selection \eqref{qbb} yields
	\begin{subequations}\label{Dbb}
		\begin{align}
			D_1(x,u)=& (k-1)x_1 - k \frac {h(x_2)}{u} + h(x_2)g(x_2) \label{D1bb}\\
			D_2(x,u)=& \left[g(x_2)-\frac{k}{u}\right]\left[-g(x_2)u(x_2) +1\right]\label{D2bb}
		\end{align}
	\end{subequations}
	
	From \eqref{D2bb} we see that $D_2(x,u)$ is independent of $x_1$, consequently the equation \eqref{poilem} of \textbf{C3} corresponds to $\nabla_2 D_1(x,u)=0$, which yields the  ODE 
	\begin{align}
		\frac{1}{k}  [g(x_2)h'(x_2)+h(x_2)]\hat u^2(x_2) - h'(x_2) \hat u(x_2) + h(x_2) \hat u'(x_2)&=0 \label{odebst}
	\end{align}
	where, to get \eqref{odebst} from \eqref{D2bb}, we employed the fact that $g'(x_2)=1$ for both converters and multiplied throughout by $\hat u^2(x_2)$. Some simple calculations show that the control law \eqref{ubst} solves the ODE \eqref{odebst}. Moreover, it is possible to verify that replacing the control in the system dynamics  fulfills \textbf{C4}. 
	
	To verify condition \textbf{C5} we compute the Hessian of the vector $\hat D(x):=D(x,\hat u(x_2)) $ to obtain
	\begin{align}\label{nDbst}
		\nabla \hat D(x) = \begin{bmatrix}k-1 & 0 \\ 0 & [ \hat D_2(x_2)]'  \end{bmatrix}.
	\end{align}
		with the $(2,2)$-element given by
\begin{align}\label{D2hes}
[ \hat D_2(x_2)]' =& -\left[g(x_2)-\frac{k}{\hat u(x_2)}\right][g(x_2) \hat u(x_2)]' \nonumber + \left[g'(x_2)+\frac{k}{\hat u^2(x_2)}\right]
\left[-g(x_2) \hat u(x_2)+1\right] \nonumber\\
=& -\left[g(x_2)-\frac{k}{\hat u(x_2)}\right][g(x_2)\hat u'(x_2) + \hat u(x_2)] 
+ \left[1+\frac{k}{\hat u^2(x_2)}\right]\left[-g(x_2) \hat u(x_2)+1\right].
\end{align}
		Positive definiteness of the Hessian at the equilibrium point---and, therefore, accomplishment of condition \textbf{C5}---is subsequently analyzed. Before proceeding, we notice  that at the equilibrium $\hat f_{1\star}=0$---see \eqref{bst}. In other words,  
		$$-g_\star \hat  u_\star+1=0.$$
		This fact reduces \eqref{D2hes} evaluated at the equilibrium to
		\begin{align*}
			\left([ \hat D_2(x_2)]'\right)_\star =&- \left[g_{\star}-\frac{k}{\hat u_\star}\right]\left[g_{\star}\hat u'_\star+ \hat u_\star\right]\\
			=&  -g_{\star}^2 u' -g_{\star}u_\star  + kg_{\star}\frac{u'_\star}{u_\star} + k \\
			=&  -\frac{u'_\star }{u_\star^2} -1 + k\frac{u'_\star}{u_\star^2} + k\\
			=&\frac{u_\star'}{u_\star^2}\left[k-1\right]+k-1
		\end{align*}
		where, to obtain the third line we use the equilibrium relation in \eqref{eqbs}
		\begin{align}\label{gstar}
			g_\star = \frac{1}{\hat u_\star}.
		\end{align}

		For the matrix $(\nabla \hat D )_\star$ to satisfy positive definiteness,  its $(1,1)$-element must be necessarily positive, that is, 
		\begin{align}
			k-1>0,	\label{condk}
		\end{align}
which is ensured in \eqref{ubstcon} holds.

		It is also required that $\mathrm{det}\left\{(\nabla \hat D)_\star\right\}>0$. Bearing that in mind, we observe, as an intermediary step, that manipulating \eqref{odebst} with  $x=x_\star$  yields
		\begin{align*}
			u'_\star&= \frac{u_\star}{h_\star}\left[ h'_\star - \frac{u_\star}{k}\left(g_\star h'_\star+h_\star \right)\right]\\
			&= \frac{u_\star}{h_\star}\left[ h'_\star -\frac{1}{k}\left( h'_\star+  h_\star u_\star\right)\right]
		\end{align*}
		where \eqref{gstar}  was newly utilized to get the last line. Hence, the determinant of \eqref{nDbst} at the equilibrium is {
		\begin{align*}
			\mathrm{det}\left\{(\nabla \hat D)_\star\right\}= & (k-1)^2\frac{u_\star'}{u_\star^2}+(k-1)^2\\\nonumber
			=& \frac{1}{h_\star u_\star}(k-1)^2\left[ h'_\star -\frac{1}{k}\left( h'_\star+  h_\star u_\star\right)\right]+(k-1)^2.
		\end{align*}}
It is worth mentioning that, in the last exression, $h'_\star>0$ (by assumption),  $u_\star>0$,  and $h_\star=x_{1\star}u_\star>0$  since $x_{1\star}$ is also positive. Thus, regarding the first term of the last determinant expression, it becomes non-negative whenever
		\begin{align*}
h'_\star \geq \frac{1}{k}\left( h'_\star+  h_\star u_\star\right)=	\frac{1}{k}\left( h'_\star+ {h_\star \over g_\star}\right),
		\end{align*} 
where we used \eqref{gstar} to obtain the right hand side identity. This is ensured if $k$ satisfies \eqref{ubstcon} verifying \textbf{C5}.
	
	The final part of this proof consists in obtaining $P(x)$. For that, we first replace $\hat u(x_2)$ into $D_1(x,u)$---see \eqref{D1bb}---as indicated below
	\begin{align*}
		D_1(x,u)=& (k-1)x_1 -k \frac{h(x_2)}{\hat u(x_2)} +h(x_2)g(x_2)  \\
		=& (k-1)x_1 - c \\
		=& \nabla_1 P(x),
	\end{align*}
	Integration of the later  equation with respect $x_1$ gives the next expression of $P(x)$ 
	$$P(x) = \frac{1}{2}(k-1)x_1^2 - cx_1 + \mu(x_2), $$ where the function $\mu(x_2)$ is the ``integration constant'' to be defined. The derivative of $P(x)$  with respect $x_2$  is then equated to $D_2(x,u)$---see \eqref{D2bb}---yielding
\begin{align}
\mu'(x_2) 
&= k g(x_2) - c \frac{1}{h(x_2)} - k \frac{g^2(x_2) h(x_2)}{g(x_2) h(x_2) + c} \nonumber\\
&= \nabla_2 P(x). \nonumber
\end{align}
	Now, $h(x_2)$ is fixed as in \eqref{hx2}. From straightforward manipulations of $\mu'(x_2)$, it follows that 
	\begin{align*}
		\mu '(x_2)   
		&=  kg(x_2) - \frac{c}{2\RR } \frac{[x_2h(x_2)]'}{x_2h(x_2)} - k\frac{g^2(x_2) h(x_2)}{g(x_2)h(x_2)+c} 
	\end{align*}

	The function $\mu(x_2)$ is hence obtained after integrating with respect $x_2$  as follows
	\begin{align*}
		\mu(x_2) = & k\int g(s)ds  - \frac{c}{2\RR} \ln (x_2h(x_2)) \nonumber
-k \int \frac{g^2(s) h(s)}{g(s)h(s)+c }\;\; ds +\text{const.}, 
	\end{align*}
	where the last term  refers to  a free constant. The Lyapunov functions for the Boost and Buck-Boost converters stated in the Lemma are obtained by replacing $g(x_2)$ for each one of the converters, evaluating the resulting integral and an appropriate selection of the constant term.

\end{document}